\documentclass[11pt,letter]{article}

\usepackage{amsmath}
\usepackage{xspace}
\usepackage{amssymb}
\usepackage{epsfig}
\usepackage{mathtools}
\usepackage{subfig}

\newtheorem{Thm}{Theorem}
\newtheorem{Lem}[Thm]{Lemma}
\newtheorem{Cor}[Thm]{Corollary}
\newtheorem{Prop}[Thm]{Proposition}

\newtheorem{Def}{Definition}

\newenvironment{proof}{\noindent {\textbf{Proof }}}{$\Box$ \medskip}

\newcommand\mbR{\mbox{$\mathbb{R}$}}

\newcommand {\ie} {\textit{i.e.}\xspace}
\newcommand {\st} {\textit{s.t.}\xspace}

\newcommand{\taba}{\hspace*{0.25in}}
\newcommand{\tabb}{\hspace*{.5in}}
\newcommand{\tabc}{\hspace*{.75in}}
\newcommand{\tabd}{\hspace*{1.0in}}

\newcommand{\tabf}{\hspace*{1.5in}}

\newcommand{\tabi}{\hspace*{2.25in}}

\title{Structure and Algorithm for Path of Solutions to a Class of Fused Lasso Problems}

\author{
Cheng Lu\\
Department of Industrial Engineering and Operations Research, \\
University of California, Berkeley \\
email: {\tt chenglu@berkeley.edu}}

\begin{document}

\maketitle

\begin{abstract}
We study a class of fused lasso problems where the estimated parameters in a sequence are regressed toward their respective observed values (fidelity loss), with $\ell_1$ norm penalty (regularization loss) on the differences between successive parameters, which promotes local constancy. In many applications, there is a coefficient, often denoted as $\lambda$, on the regularization term, which adjusts the relative importance between the two losses.


In this paper, we characterize how the optimal solution evolves with the increment of $\lambda$. We show that, if all fidelity loss functions are convex piecewise linear, the optimal value for \emph{each} variable changes at most $O(nq)$ times for a problem of $n$ variables and total $q$ breakpoints. On the other hand, we present an algorithm that solves the path of solutions of \emph{all} variables in $\tilde{O}(nq)$ time for all $\lambda \geq 0$. Interestingly, we find that the path of solutions for each variable can be divided into up to $n$ locally convex-like segments. For problems of arbitrary convex loss functions, for a given solution accuracy, one can transform the loss functions into convex piecewise linear functions and apply the above results, giving pseudo-polynomial bounds as $q$ becomes a pseudo-polynomial quantity.


To our knowledge, this is the first work to solve the path of solutions for fused lasso of non-quadratic fidelity loss functions.

\end{abstract}

\section{Introduction}
In this paper, we characterize and solve the path of solutions to the following class of fused lasso problems:
\begin{equation}\label{prob: fused-lasso-lambda}
\begin{split}
(\text{FL})\ \min_{x_1,\ldots, x_n}\ &\sum_{i=1}^n f_i(x_i) + \lambda\sum_{i=1}^{n-1}|x_i - x_{i+1}|\\
\st\ &\ell_i \leq x_i \leq u_i,\ i = 1,\ldots,n.
\end{split}
\end{equation}
Each function $f_i(x_i)$ is a general convex function. The coefficient $\lambda$ ($\lambda \geq 0$) is a hyperparameter for the problem. The optimal solution varies with regard to $\lambda$. Let $\boldsymbol{x}^*(\lambda) = (x^*_1(\lambda), x^*_2(\lambda), \ldots, x^*_n(\lambda))$ be the optimal solution to FL for a given $\lambda$. The optimal solution $\boldsymbol{x}^*(\lambda)$ is a function of $\lambda$ and we refer the function, with $\lambda$ defined over $[0, +\infty)$, as the \emph{path of solutions} to problem FL. Without loss of generality, we only consider integer $\lambda$ values.

We first characterize and solve the path of solutions to a special case of FL:
\begin{equation*}
\begin{split}
(\text{PL-FL})\ \min_{x_1,\ldots, x_n}\ &\sum_{i=1}^n f^{pl}_i(x_i) + \lambda\sum_{i=1}^{n-1}|x_i - x_{i+1}|\\
\st\ &\ell_i \leq x_i \leq u_i,\ i = 1,\ldots,n.
\end{split}
\end{equation*}
Each function $f^{pl}_i(x_i)$ is a convex piecewise linear function (the superscript ``pl" stands for ``piecewise linear") of $q_i$ breakpoints. Note that any FL problem can be ``piecewise linearized" to a PL-FL problem for a given solution accuracy $\epsilon$ \cite{HS90}, where we solve an $\epsilon$-accurate solution $\boldsymbol{x}$ for FL such that there is an optimal solution $\boldsymbol{x}^*$ for FL satisfying $||\boldsymbol{x} - \boldsymbol{x}^*||_{\infty} < \epsilon$. In other words, $\boldsymbol{x}$'s first $\log\frac{1}{\epsilon}$ significant digits after the decimal point are identical to those of $\boldsymbol{x}^*$. The ``piecewise linearization" is done by introducing $q_i = (u_i - \ell_i)/\epsilon + 1$ breakpoints for each $f_i(x_i)$, $\{\ell_i, \ell_i + \epsilon, \ell_i + 2\epsilon, \ldots, u_i - 2\epsilon, u_i - \epsilon, u_i\}$, and defining a convex piecewise linear function $f^{pl}_i$ whose left and right sub-gradients, if exists, on each breakpoint are defined as:
\begin{equation*}
\begin{split}
(f^{pl}_i)'_L(x) &= (f_i(x) - f_i(x-\epsilon))/\epsilon,\\
(f^{pl}_i)'_R(x) &= (f_i(x+\epsilon) - f_i(x))/\epsilon.
\end{split}
\end{equation*}
With the transformation, a bound for PL-FL directly leads to a bound for FL. The caveat is that, while $q_i$ is an input parameter for PL-FL\footnote{We assume that a piecewise linear function is represented by a sorted list of breakpoints with slopes of linear pieces in-between}, it is not for FL. As a result, a bound for PL-FL that is polynomial of $q_i$ becomes a pseudo-polynomial bound for FL.

Without loss of generality, any convex piecewise linear function $f^{pl}_i(x_i)$ with box constraint $\ell_i \leq x_i \leq u_i$ is equivalent to a convex piecewise linear function without the box constraint:
\begin{equation*}
\tilde{f}^{pl}_i(x_i) = \begin{cases}
				f^{pl}_i(\ell_i) - M(x_i - \ell_i) \ \text{for}\ x_i < \ell_i,\\
				f^{pl}_i(x_i) \ \text{for}\ \ell_i \leq x_i \leq u_i,\\
				f^{pl}_i(u_i) + M(x_i - u_i) \ \text{for}\ x_i > u_i	
				\end{cases}
\end{equation*}
for $M$ sufficiently large. Therefore, in the paper, the PL-FL problem is unconstraint, without loss of generality:
\begin{equation}\label{prob: PL-FL-lambda}
(\text{PL-FL})\ \min_{x_1,\ldots, x_n}\ \sum_{i=1}^n f^{pl}_i(x_i) + \lambda\sum_{i=1}^{n-1}|x_i - x_{i+1}|,
\end{equation}
with the first piece of each convex piecewise linear function having negative slope $(-M)$ and the last piece of each convex piecewise linear function having positive slope $(+M)$. To simplify notation, the path of solutions to PL-FL (\ref{prob: PL-FL-lambda}) is also denoted as $\boldsymbol{x}^*(\lambda) = (x^*_1(\lambda), x^*_2(\lambda), \ldots, x^*_n(\lambda))$. In the remainder of the paper, it should be clear which problem an optimal solution refers to from the context.

In this paper, we show that, for PL-FL (\ref{prob: PL-FL-lambda}), the optimal value for \emph{each} variable $x_i$ changes at most $O(nq)$ times as $\lambda$ increases, where $q = \sum_{i=1}^n q_i$ is the total number of breakpoints (counting multiplicity) of all the $n$ convex piecewise linear functions. On the other hand, we present an algorithm that solves the path of solutions of \emph{all} variables in $\tilde{O}(nq)$ time. The two bounds only differ by a logarithmic factor. In addition, we find that the path of solutions for each variable can be divided into up to $n$ locally convex-like segments.

With the above transformation between FL (\ref{prob: fused-lasso-lambda}) and PL-FL (\ref{prob: PL-FL-lambda}), we have $q = O(\frac{nU}{\epsilon})$ for transformed FL (\ref{prob: fused-lasso-lambda}) of solution accuracy $\epsilon$, where $U = \max_i\{u_i - \ell_i\}$. As a result, applying the above bounds, we have, in FL (\ref{prob: fused-lasso-lambda}) of solution accuracy $\epsilon$, the optimal value for \emph{each} variable $x_i$ changes at most $O(\frac{n^2U}{\epsilon})$ times as $\lambda$ increases, and the path-of-solution algorithm has time complexity $\tilde{O}(\frac{n^2U}{\epsilon})$.

\subsection{Applications of PL-FL}
Besides being a bridge for FL (\ref{prob: fused-lasso-lambda}), special cases of PL-FL problem (\ref{prob: PL-FL-lambda}) appear in many applications. An example is in array-CGH analysis in bioinformatics \cite{EdM05}. It is to estimate the ratio of gene copying numbers at each position in DNA sequences between tumor and normal cell samples, based on the biological knowledge that the ratios between adjacent positions in the DNA sequences are similar. Eilers and de Menezes in \cite{EdM05} proposed the following quantile fused lasso model to identify the estimated log-ratio $x_i$, based on the observed log-ratio $a_i$ at the $i$th position:
\begin{equation*}
\min_{x_1,\ldots,x_n}\ \sum_{i=1}^n\rho_\tau(x_i; a_i) + \lambda\sum_{i=1}^{n-1}|x_i - x_{i+1}|,
\end{equation*}
where $\rho_\tau(x_i; a_i)$ is a quantile function defined for parameter $\tau \in [0,1]$ as:
\begin{equation*}
\rho_\tau(x_i; a_i) = \begin{cases}
				\tau(x_i - a_i)\ \text{if}\ x_i - a_i \geq 0,\\
				-(1-\tau)(x_i - a_i)\ \text{if}\ x_i - a_i < 0.
				\end{cases}
\end{equation*}

In signal processing, Storath, Weinmann, and Unser in \cite{SWU16} considered a fused lasso model with $\ell_1$ fidelity loss functions:
\begin{equation*}
\min_{x_1,\ldots,x_n}\ \sum_{i=1}^nw_i|x_i - a_i| + \lambda\sum_{i=1}^{n-1}|x_i - x_{i+1}|,
\end{equation*}
where the $w_i$'s are positive weights.

In the above models, the hyperparameter $\lambda$ weights the relative importance between the fidelity loss and the regularization loss. It is often selected by solving the problem for many different values of $\lambda$ and choose the best one by examining the respective optimal solutions. This is often time and labor consuming. As we shall show, if the set/interval of candidate $\lambda$ values is large, our path-of-solution algorithm is faster than solving the problem for each candidate $\lambda$ from scratch.
 
\subsection{Existing path-of-solution algorithms}
Existing works on the solution path of special cases and variants of FL problem (\ref{prob: fused-lasso-lambda}) inspire the work in the paper. A special case of FL (\ref{prob: fused-lasso-lambda}), called fused-lasso signal approximator (FLSA), is studied in \cite{FHHT07, HH10}. The problem is defined as follows:
\begin{equation}\label{prob: FLSA}
(\text{FLSA})\ \min_{x_1,\ldots, x_n}\ \frac{1}{2}\sum_{i=1}^n(x_i - a_i)^2 + \lambda_1\sum_{i=1}^n|x_i| + \lambda_2\sum_{i=1}^{n-1}|x_i - x_{i+1}|.
\end{equation}
Friedman et al. in \cite{FHHT07} prove a ``fusing property" of FLSA (\ref{prob: FLSA}). Let $\lambda_1$ be fixed. They prove that if the optimal values of $x_i$ and $x_{i+1}$ are equal for a $\lambda_2$, then for all $\lambda'_2 > \lambda_2$, the optimal values of $x_i$ and $x_{i+1}$ remain equal. Inspired by their proof technique, we shall prove in the paper that the same fusing property holds for FL (\ref{prob: fused-lasso-lambda}) of arbitrary convex fidelity loss functions, not only the convex quadratic-type functions in FLSA (\ref{prob: FLSA}). Hoefling in \cite{HH10} provides an efficient path-of-solution algorithm to solve FLSA (\ref{prob: FLSA}) for all values of $\lambda_1 \geq 0$ and $\lambda_2 \geq 0$. Hoefling's algorithm has time complexity $O(n\log n)$, and the space complexity to store the path of solutions is $O(n)$.

Tibshirani and Taylor in \cite{TT11} present path-of-solution algorithms to a generalized lasso problem as follows:
\begin{equation}
(\text{Generalized Lasso})\ \min_{\boldsymbol{x}\in \mbR^n}\ \frac{1}{2}||\boldsymbol{b} - A\boldsymbol{x}||^2_2 + \lambda||D\boldsymbol{x}||_1. 
\end{equation}
The case of interest here is $A = I$ and $D$ being the 1-dimensional fused lasso matrix, which is also a special case of FL (\ref{prob: fused-lasso-lambda}). A path-of-solution algorithm for this case is discussed in Section 5 of \cite{TT11}, yet the time complexity of the algorithm and the space complexity to store the path of solutions are not explicitly provided. Note that this case is also a special case of FLSA (\ref{prob: FLSA}) with $\lambda_1 = 0$.

As a variant, Tibshirani et al. in \cite{THT11} present an efficient path-of-solution algorithm to solve a ``nearly-isotonic" problem:
\begin{equation}
(\text{Nearly-isotonic})\ \min_{x_1,\ldots, x_n}\ \frac{1}{2}\sum_{i=1}^n(x_i - a_i)^2 + \lambda\sum_{i=1}^{n-1}(x_i - x_{i+1})_+,
\end{equation}
where $(x)_+$ is the positive part of $x$, $\max\{x,0\}$. The path-of-solution algorithm in \cite{THT11} for the nearly-isotonic problem has $O(n\log n)$ time complexity and $O(n)$ space complexity to store the path of solutions.

To our knowledge, the work presented here is the first to solve the path of solutions for fused lasso problems of non-quadratic fidelity loss functions.

\subsection{Overview}
The rest of the section is organized as follows. In Section \ref{sect: bound-fusing-breakpoints}, we prove that for FL (\ref{prob: fused-lasso-lambda}) of arbitrary convex loss functions, once a pair of neighboring variables are fused together for a $\lambda$ value, then for any $\lambda' > \lambda$, the pair of neighboring variables remain fused. This result leads to the definition of fusing $\lambda$ values and bounds the number of fusing $\lambda$ values by $O(n)$. The fusing $\lambda$ values partition the whole interval $[0,+\infty)$ into $O(n)$ segments such that in each segment, no variables are fused together. With this observation, in Section \ref{sect: bound-between-fusing-values}, we bound the number of different solutions to PL-FL (\ref{prob: PL-FL-lambda}) and FL (\ref{prob: fused-lasso-lambda}) between two adjacent fusing $\lambda$ values. The above two sections together bound the number of times a variable changes its optimal value and characterize how the path of solutions look like as $\lambda$ increases.

Inspired by the proofs in Section \ref{sect: bound-fusing-breakpoints} and \ref{sect: bound-between-fusing-values}, an algorithm is designed in Section \ref{sect: path_alg} to solve the path of solutions to PL-FL (\ref{prob: PL-FL-lambda}) in polynomial time and FL (\ref{prob: fused-lasso-lambda}) in pseudo-polynomial time. Concluding remarks are provided in Section \ref{sect: conclusion}.

\section{Bounding the number of fusing $\lambda$ values}\label{sect: bound-fusing-breakpoints}
We formally define the concepts of a \emph{fusing} $\lambda$ value as follows:
\begin{Def}
$\lambda_0$ is a \emph{fusing} $\lambda$ value for FL (\ref{prob: fused-lasso-lambda}) if there exists an $i$ such that $x^*_{i}(\lambda) \neq x^*_{i+1}(\lambda)$ for any $\lambda < \lambda_0$ but $x^*_{i}(\lambda) = x^*_{i+1}(\lambda)$ for all $\lambda \geq \lambda_0$.
\end{Def}
The validity of the above definition is supported by the following theorem:
\begin{Thm}\label{thm: fusing-points}
For FL (\ref{prob: fused-lasso-lambda}), suppose $x^*_{i}(\lambda_0) = x^*_{i+1}(\lambda_0)$ for two adjacent coordinates $i$ and $i+1$ for some $\lambda_0 \geq 0$, then for any $\lambda \geq \lambda_0$, we have $x^*_{i}(\lambda) = x^*_{i+1}(\lambda)$.
\end{Thm}
\begin{proof}
At $\lambda_0$, suppose we have a stretch of joined coordinates $j, j+1, \ldots, j+k$ that include $i$ and $i+1$ such that $x^* = x^*_{j}(\lambda_0) = x^*_{j+1}(\lambda_0) = \ldots = x^*_{j+k}(\lambda_0)$. The certificate for a solution to be optimal for FL (\ref{prob: fused-lasso-lambda}) is that the set of (partial) sub-gradients with regard to each coordinate contains 0. Thus, for $\lambda_0$, there exist values of $s_{j-1}, s_j, \ldots, s_{j+k}$ such that, together with $x^*$, make the following sub-gradient equations hold:
\begin{equation}\label{eqn: subgradients}
\partial f_{i'}(x) + \lambda(s_{i'} - s_{i'-1}) = 0\ (i' = j,\ldots,j+k)
\end{equation}
where $\partial f_{i'}(x)$ is a sub-gradient value of $f_{i'}$ at $x$, and $s_{i'}$ has the constraints that:
\begin{equation}\label{eqn: s_constraints}
\begin{cases}
s_{i'} = 1,\ \text{if}\ x_{i'}(\lambda_0) > x_{i'+1}(\lambda_0) \\
s_{i'} \in [-1,1],\ \text{if}\ x_{i'}(\lambda_0) = x_{i'+1}(\lambda_0) \\
s_{i'} = -1,\ \text{if}\ x_{i'}(\lambda_0) < x_{i'+1}(\lambda_0).
\end{cases}
\end{equation}
For notation convenience, we let $s_0 = s_n = 0$.

Summing up the equations in (\ref{eqn: subgradients}), we have
\begin{equation}\label{eqn: sum_subgradients}
\sum_{i'=j}^{j+k}\partial f_{i'}(x) + \lambda(s_{j+k} - s_{j-1}) = 0.
\end{equation}
Note that $s_{j+k}, s_{j-1} \in \{-1, 1\}$, and they remain constants as long as the group of coordinates $\{j, j+1,\ldots, j+k\}$ do not merge with the adjacent ones.

On the other hand, taking pairwise differences of equations in (\ref{eqn: subgradients}), we have:
\begin{equation}\label{eqn: succ-diff}
\partial f_{i'+1}(x) - \partial f_{i'}(x) + \lambda(s_{i'+1} - 2s_{i'} + s_{i'-1}) = 0\ (i' = j,\ldots,j+k-1).
\end{equation}
The values of $s_j, s_{j+1}, \ldots, s_{j+k-1}$ satisfy the following equations:
\begin{equation}\label{eqn: succ-diff-mat}
\boldsymbol{A}s = \frac{1}{\lambda}\Delta f + c,
\end{equation}
where
\begin{align*}
&\boldsymbol{A} = \begin{bmatrix}
2 & -1 & 0 & \ldots & 0 & 0 & 0 \\
-1 & 2 & -1 & \ldots & 0 & 0 & 0 \\
\ldots \\
0 & 0 & 0 &\ldots & 0 & -1 & 2
\end{bmatrix},\\
&s = (s_j, \ldots, s_{j+k-1})^T,\\
&\Delta f = (\partial f_{j+1}(x) - \partial f_{j}(x), \ldots, \partial f_{j+k}(x) - \partial f_{j+k-1}(x))^T,\\
&c = (s_{j-1}, 0, \ldots, 0, s_{j+k}).
\end{align*}
Since $\boldsymbol{A}$ is invertible, for any value $\lambda \geq \lambda_0$, the value of $s$ is uniquely determined by $\Delta f$ ($c$ is constant). And the value of $s$ satisfy the constraints (\ref{eqn: s_constraints}) if and only if the elements of $\Delta f$ are in the following range:
\begin{equation}\label{eqn: succ-diff_range}
\begin{cases}
\partial f_{j+1}(x) - \partial f_j(x) \in [\lambda(-3-s_{j-1}), \lambda(3-s_{j-1})]\\
\partial f_{i'+1}(x) - \partial f_{i'}(x) \in [-4\lambda, 4\lambda]\ (i' = j+1,\ldots,j+k-2)\\
\partial f_{j+k}(x) - \partial f_{j+k-1}(x) \in [\lambda(-3-s_{j+k}), \lambda(3-s_{j+k})].
\end{cases}
\end{equation}

We consider three cases depending on the values of $s_{j+k}$ and $s_{j-1}$:
\begin{enumerate}
\item $s_{j+k} - s_{j-1} = 0$. This includes the cases $s_{j+k} = s_{j-1} \in \{-1, 0, 1\}$ (where the case of equaling to $0$ is for the boundary case of $j = 1$ and $j + k = n$). W.o.l.g., we assume $s_{j+k} = s_{j-1} = 1$. For any $\lambda \geq \lambda_0$, the equation (\ref{eqn: sum_subgradients}) holds true for $x^*$. And $x^*$ satisfies
\begin{equation*}
\begin{cases}
\partial f_{j+1}(x^*) - \partial f_j(x^*) \in [-4\lambda_0, 2\lambda_0] \subseteq [-4\lambda, 2\lambda]\\
\partial f_{i'+1}(x^*) - \partial f_{i'}(x^*) \in [-4\lambda_0, 4\lambda_0] \subseteq [-4\lambda, 4\lambda]\ (i' = j+1,\ldots,j+k-2)\\
\partial f_{j+k}(x^*) - \partial f_{j+k-1}(x^*) \in [-4\lambda_0, 2\lambda_0] \subseteq [-4\lambda, 2\lambda].
\end{cases}
\end{equation*}
Hence $x^*$ remains an optimal solution for variables $x_j, \ldots, x_{j+k}$ for $\lambda \geq \lambda_0$.

\item $s_{j+k} - s_{j-1} = 2$. This corresponds to the case where $s_{j+k} = 1$ and $s_{j-1} = -1$. As $\lambda$ increases from $\lambda_0$, the term $\sum_{i'=j}^{j+k}\partial f_{i'}(x^*)$ should decrease in order to satisfy equation (\ref{eqn: sum_subgradients}). We prove that there exists a value $x^*_{\ell} \leq x^*$ that is optimal for variables $x_j, \ldots, x_{j+k}$ for $\lambda \geq \lambda_0$. This is equivalent to proving that
    \begin{equation}\label{eqn: succ-diff_equiv}
    \begin{cases}
    \partial f_{j+1}(x^*_\ell) - \partial f_j(x^*_\ell) \in [-2\lambda, 4\lambda]\\
    \partial f_{i'+1}(x^*_\ell) - \partial f_{i'}(x^*_\ell) \in [-4\lambda, 4\lambda]\ (i' = j+1,\ldots,j+k-2)\\
    \partial f_{j+k}(x^*_\ell) - \partial f_{j+k-1}(x^*_\ell) \in [-4\lambda, 2\lambda]
    \end{cases}
    \end{equation}

By equation (\ref{eqn: sum_subgradients}), the value of $x^*_\ell$ satisfies $\sum_{i'=j}^{j+k}\partial f_{i'}(x^*) - \sum_{i'=j}^{j+k}\partial f_{i'}(x^*_\ell) = 2(\lambda - \lambda_0)$. One can rewrite the left hand side of the relations in (\ref{eqn: succ-diff_equiv}) as
    \begin{equation}\label{eqn: partition-subtraction}
    \begin{split}
    &\quad \partial f_{i'+1}(x^*_\ell) - \partial f_{i'}(x^*_\ell) \\
    &= (\partial f_{i'+1}(x^*) - \partial f_{i'}(x^*)) + (\partial f_{i'+1}(x^*_\ell) - \partial f_{i'+1}(x^*)) - (\partial f_{i'}(x^*_\ell) - \partial f_{i'}(x^*))\ (i' = j,\ldots,j+k-1).
    \end{split}
    \end{equation}
    As $x^*$ is optimal for $\lambda_0$, we have
    \begin{equation}\label{eqn: succ-diff_equiv1}
    \begin{cases}
    \partial f_{j+1}(x^*) - \partial f_j(x^*) \in [-2\lambda_0, 4\lambda_0]\\
    \partial f_{i'+1}(x^*) - \partial f_{i'}(x^*) \in [-4\lambda_0, 4\lambda_0]\ (i' = j+1,\ldots,j+k-2)\\
    \partial f_{j+k}(x^*) - \partial f_{j+k-1}(x^*) \in [-4\lambda_0, 2\lambda_0].
    \end{cases}
    \end{equation}
    On the other hand, the convexity of $f_{i'}$ functions implies that
    \begin{equation*}
    0 \leq \partial f_{i'}(x^*) - \partial f_{i'}(x^*_\ell) \leq 2(\lambda - \lambda_0)\ (i' = j,\ldots,j+k).
    \end{equation*}
    Thus we have
    \begin{equation}\label{eqn: succ-diff_equiv2}
    (\partial f_{i'+1}(x^*_\ell) - \partial f_{i'+1}(x^*)) - (\partial f_{i'}(x^*_\ell) - \partial f_{i'}(x^*)) \in [-2(\lambda-\lambda_0), 2(\lambda-\lambda_0)]\ (i' = j, \ldots, j+k-1).
    \end{equation}
    Finally we add the inclusion relations (\ref{eqn: succ-diff_equiv1}) and (\ref{eqn: succ-diff_equiv2}) to (\ref{eqn: partition-subtraction}), implying that (\ref{eqn: succ-diff_equiv}) hold.

\item $s_{j+k} - s_{j-1} = -2$. This corresponds to the case where $s_{j+k} = -1$ and $s_{j-1} = 1$. This case is symmetric to case 2, and thus can be proved that there exists an $x^*_u \geq x^*$ that is optimal for variables $x_j, \ldots, x_{j+k}$ for $\lambda \geq \lambda_0$.
\end{enumerate}
\end{proof}

As there are $n$ coordinates in FL (\ref{prob: fused-lasso-lambda}), we immediately have
\begin{Cor}\label{cor: fusing-num}
The number of fusing $\lambda$ values for FL (\ref{prob: fused-lasso-lambda}) is at most $n-1$.
\end{Cor}

\section{Bound the number of different solutions between two adjacent fusing $\lambda$ values}\label{sect: bound-between-fusing-values}
Given the concept of fusing $\lambda$ values, for any $\lambda \geq 0$, an optimal solution to FL (\ref{prob: fused-lasso-lambda}) can be partitioned into groups of adjacent coordinates, where variables in a same group have identical optimal value. The $O(n)$ fusing $\lambda$ values act as anchors on the interval $[0, +\infty)$ that cut the interval into $O(n)$ sub-intervals. Inside each sub-interval, as $\lambda$ varies, the group partition is not changed, yet the exact identical optimal values for each group. In this section we provide uniform bounds on the number of different optimal solutions for PL-FL (\ref{prob: PL-FL-lambda}) and FL (\ref{prob: fused-lasso-lambda}) as $\lambda$ increases in any sub-interval. These bounds multiplied by $O(n)$ would bound the total number of different solutions to both problems for $\lambda \in [0, +\infty)$.

\subsection{PL-FL (\ref{prob: PL-FL-lambda})}

The characterization of the bound on the number of different solutions for PL-FL (\ref{prob: PL-FL-lambda}) is inspired by the algorithm in \cite{HL17}, where an efficient algorithm is presented to solve a generalization of PL-FL, called \emph{Generalized Isotonic Median Regression} (GIMR) problem \cite{HL17}:
\begin{equation}\label{prob: GIMR}
(\text{GIMR})\ \min\ \sum_{i=1}^n f^{pl}_i(x_i) + \sum_{i=1}^{n-1}d_{i,i+1}(x_i - x_{i+1})_+ + \sum_{i=1}^{n-1}d_{i+1,i}(x_{i+1} - x_i)_+.
\end{equation}
$(x)_+ = x$ if $x > 0$ and $0$ otherwise. The $d_{i,i+1}$ and $d_{i+1,i}$ are fixed nonnegative coefficients. GIMR generalizes PL-FL in that the absolute difference $|x_i - x_{i+1}|$ is split into two terms, each with different coefficients. Hochbaum and Lu in \cite{HL17} give an efficient $O(q\log n)$ algorithm (called HL-algorithm hereafter) to solve GIMR (\ref{prob: GIMR}) for any given $d_{i,i+1}$ and $d_{i+1,i}$.

\subsubsection{Overview of HL-algorithm for GIMR}\label{sect: recap}
In this section, we give an overview of the algorithm for GIMR (\ref{prob: GIMR}) in \cite{HL17}. We first introduce the notation and preliminaries necessary to present the algorithm. These notation and preliminaries are used throughout the paper. Key results of HL-algorithm in \cite{HL17} then follow.

\paragraph{Notation and Preliminaries}
GIMR (\ref{prob: GIMR}) can be viewed as defined on a bi-directional path (bi-path) graph $G = (V, A)$ with node set $V = \{1,2,\ldots,n\}$ and $A = \{(i,i+1), (i+1,i)\}_{i=1,\ldots,n-1}$. Each node $i$ in the graph corresponds to variable $x_i$.

Let \emph{interval} $[i,j]$ in bi-path graph $G$ for $i \leq j$ be the subset of $V$, $\{i, i+1,\ldots, j-1,j\}$. If $i = j$, the interval $[i,i]$ is the singleton $i$. The notations $[i,j)$ and $(i,j]$ indicate the intervals $[i,j-1]$ and $[i+1,j]$ respectively. Let $[i,j] = \emptyset$ if $i > j$.

Let the directed $s,t$-graph $G^{st} = (V_{st}, A_{st})$ be associated with graph $G = (V,A)$ such that $V_{st} = V\cup\{s,t\}$ and $A_{st} = A\cup A_s \cup A_t$. The appended node $s$ is called the \emph{source} node and $t$ is called the \emph{sink} node. $A_s = \{(s,i): i\in V\}$ and $A_t = \{(i,t): i\in V\}$ are the respective sets of \emph{source adjacent arcs} and \emph{sink adjacent arcs}. Each arc $(i,j)\in A_{st}$ has an associated nonnegative capacity $c_{i,j}$.

For any two subsets of nodes $V_1,V_2 \subseteq V_{st}$, we let $(V_1, V_2) = \{(i,j)\in A_{st}|i\in V_1, j\in V_2\}$ and $C(V_1, V_2) = \sum_{(i,j)\in (V_1, V_2)}c_{i,j}$.

An \emph{$s,t$-cut} is a partition of $V_{st}$, $(\{s\}\cup S, T\cup \{t\})$, where $T = \bar{S} = V\setminus S$. For simplicity, we refer to an $s,t$-cut partition as $(S,T)$. We refer to $S$ as the \emph{source set} of the cut, excluding $s$. For each node $i\in V$, we define its \emph{status} in graph $G^{st}$ as $status(i) = s$ if $i \in S$ (referred as an \emph{$s$-node}), otherwise $status(i) = t)$ ($i \in T$) (referred as a \emph{$t$-node}).

The \emph{capacity} of a cut $(S,T)$ is defined as $C(\{s\}\cup S,T\cup\{t\})$. A \emph{minimum cut} in $s,t$-graph $G^{st}$ is an $s,t$-cut $(S,T)$ that minimizes $C(\{s\}\cup S,T\cup\{t\})$. Hereafter, any reference to a minimum cut is to the unique minimum $s,t$-cut with the \emph{maximal source set}. That means if there are multiple minimum cuts, then the one selected has a source set that is not contained in any other source set of a minimum cut.

A convex piecewise linear function $f^{pl}_i(x_i)$ is specified by its ascending list  of $q_i$ breakpoints, $a_{i,1} < a_{i,2} < \ldots < a_{i,q_i}$, and the slopes of the $q_i + 1$ linear pieces between every two adjacent breakpoints, denoted by $w_{i,0} < w_{i,1} < \ldots < w_{i,q_i}$.
Let the sorted list of the union of $q$ breakpoints of all the $n$ convex piecewise linear functions be $a_{i_1,j_1} < a_{i_2,j_2} < \ldots < a_{i_q,j_q}$ (w.l.o.g. we assume that the $n$ sets of breakpoints are disjoint, explained in \cite{HL17}), where $a_{i_k,j_k}$, the $k$th breakpoint in the sorted list, is the breakpoint between the $(j_k - 1)$th and the $j_k$th linear pieces of function $f^{pl}_{i_k}(x_{i_k})$.

\paragraph{Algorithm overview}

We construct a \emph{parametric graph} $G^{st}(\alpha) = (V_{st}, A_{st})$ associated with the bi-path graph $G = (V,A)$, for any scalar value $\alpha$. The capacities of arcs $(i,i+1), (i+1,i) \in A$ are $c_{i,i+1} = d_{i,i+1}$ and $c_{i+1,i} = d_{i+1,i}$ respectively. Each arc in $A_s = \{(s,i)\}_{i\in V}$ has capacity $c_{s,i} = \max\{0, -(f^{pl}_i)'(\alpha)\}$ and each arc in $A_t = \{(i,t)\}_{i \in V}$ has capacity $c_{i,t} = \max\{0, (f^{pl}_i)'(\alpha)\}$, where $(f^{pl}_i)'(\alpha)$ is the right sub-gradient of function $f^{pl}_i(\cdot)$ at argument $\alpha$. (One can select instead the left sub-gradient.) Note that for any given value of $\alpha$, either $c_{s,i} = 0$ or $c_{i,t} = 0$.

The link between the minimum cut for any given value of $\alpha$ and the optimal solution to GIMR (\ref{prob: GIMR}) is characterized in the following \emph{threshold theorem} \cite{DH01, HL17}:
\begin{Thm}\label{thm: threshold}\emph{(threshold theorem, Hochbaum \cite{DH01})}.
For any given $\alpha$, let $S^*$ be the maximal source set of the minimum cut in graph $G^{st}(\alpha)$. Then there is an optimal solution $\textbf{x}^*$ to GIMR (\ref{prob: GIMR}) satisfying $x^*_i \geq \alpha$ if $i \in S^*$ and $x^*_i < \alpha$ if $i \in T^*$.
\end{Thm}

An important property of $G^{st}(\alpha)$ is that the capacities of source adjacent arcs are nonincreasing functions of $\alpha$, the capacities of sink adjacent arcs are nondecreasing functions of $\alpha$, and the capacities of all the other arcs are constants. This implies the following \emph{nested cut property}:
\begin{Lem}\label{lem: nestedness}\emph{(nested cut property \cite{GGT89, DH01, DH08})}.
For any two parameter values $\alpha_1 \leq \alpha_2$, let $S_{\alpha_1}$ and $S_{\alpha_2}$ be the respective maximal source set of the minimum cuts of $G^{st}(\alpha_1)$ and $G^{st}(\alpha_2)$, then $S_{\alpha_1} \supseteq S_{\alpha_2}$.
\end{Lem}
We remark that the above threshold theorem and nested cut property both work not only for GIMR (\ref{prob: GIMR}) defined on a bi-path graph, but also for an generalization of GIMR that is defined on arbitrary (directed) graphs.

Based on the threshold theorem, it is sufficient to solve the minimum cuts in the parametric graph $G^{st}(\alpha)$ for all values of $\alpha$, in order to solve GIMR (\ref{prob: GIMR}). In piecewise linear functions, the right sub-gradients for $\alpha$ values between any two adjacent breakpoints are constant. Thus the source and sink adjacent arc capacities remain constant for $\alpha$ between any two adjacent breakpoint values in the sorted list of breakpoints over all the $n$ convex piecewise linear functions. Therefore the minimum cuts in $G^{st}(\alpha)$ remain unchanged as capacities of all the arcs in the parametric graph are unchanged. Thus we have:
\begin{Lem}
The minimum cuts in $G^{st}(\alpha)$ remain unchanged for $\alpha$ assuming any value between any two adjacent breakpoints in the sorted list of breakpoints of all the $n$ convex piecewise linear functions, $\{f^{pl}_i(x_i)\}_{i=1,\ldots,n}$.
\end{Lem}
Thus the values of $\alpha$ to be considered can be restricted to the set of breakpoints of the $n$ convex piecewise linear functions, $\{f^{pl}_i(x_i)\}_{i=1,\ldots,n}$. The HL-algorithm solves GIMR (\ref{prob: GIMR}) by efficiently computing the minimum cuts of $G^{st}(\alpha)$ for subsequent values of $\alpha$ in the ascending list of breakpoints, $a_{i_1,j_1} < a_{i_2,j_2} < \ldots < a_{i_q,j_q}$.

Let $G_k$, for $k \geq 1$, be the parametric graph $G^{st}(\alpha)$ for $\alpha$ equal to $a_{i_k,j_k}$, \ie, $G_k = G^{st}(a_{i_k,j_k})$. For $k = 0$, we let $G_0 = G^{st}(a_{i_1,j_1} - \epsilon)$ for a small value of $\epsilon > 0$. Let $(S_k, T_k)$ be the minimum cut in $G_k$, for $k \geq 0$. Recall that $S_k$ is the maximal source set. The nested cut property (Lemma \ref{lem: nestedness}) implies that $S_k \supseteq S_{k+1}$ for $k \geq 0$. Based on the threshold theorem and the nested cut property, we know that for each node $j = 1,\ldots,n$, $x^*_j = a_{i_k, j_k}$ for the index $k$ such that $j \in S_{k-1}$ and $j \in T_k$.

The HL-algorithm generates the respective minimum cuts of graphs $G_k$ in increasing order of $k$. It is shown in \cite{HL17} that $(S_k, T_k)$ can be computed from $(S_{k-1}, T_{k-1})$ in time $O(\log n)$. Hence the total complexity of the algorithm is $O(q\log n)$. The efficiency of updating $(S_k, T_k)$ from $(S_{k-1}, T_{k-1})$ is based on the following key results.

The update of the graph from $G_{k-1}$ to $G_k$ is simple as it only involves a change in the capacities of the source and sink adjacent arcs of $i_k$, $(s,i_k)$ and $(i_k,t)$. Recall that from $G_{k-1}$ to $G_k$, the right sub-gradient of $f^{pl}_{i_k}$ changes from $w_{i_k,j_k-1}$ to $w_{i_k,j_k}$. Thus the change of $c_{s,i_k}$ and $c_{i_k,t}$ from $G_{k-1}$ to $G_k$ depends on the signs of $w_{i_k, j_k-1}$ and $w_{i_k,j_k}$. There are three possible cases:\\
Case 1. $w_{i_k, j_k-1} \leq 0$, $w_{i_k, j_k} \leq 0$: $c_{s,i_k}$ is changed from $-w_{i_k,j_k-1}$ to $-w_{i_k,j_k}$.\\
Case 2. $w_{i_k, j_k-1} \leq 0$, $w_{i_k,j_k} \geq 0$: $c_{s,i_k}$ is changed from $-w_{i_k,j_k-1}$ to 0 and $c_{i_k,t}$ is changed from 0 to $w_{i_k,j_k}$.\\
Case 3. $w_{i_k, j_k-1} \geq 0$, $w_{i_k, j_k} \geq 0$: $c_{i_k,t}$ is updated from $w_{i_k,j_k-1}$ to $w_{i_k, j_k}$.\\
Note that the update from $G_{k-1}$ to $G_k$ does not involve the values of $d_{i,i+1}$ and $d_{i+1,i}$ in GIMR (\ref{prob: GIMR}), thus nor does it involve the value of $\lambda$ in PL-FL (\ref{prob: PL-FL-lambda}).

Based on the nested cut property, for any node $i$, if $i \in T_{k-1}$, then $i$ remains in the sink set for all subsequent cuts, and in particular $i \in T_k$. Hence an update of the minimum cut in $G_k$ from the minimum cut in $G_{k-1}$ can only involve shifting some nodes from source set $S_{k-1}$ to sink set $T_k$. Formally, the relation between $(S_{k-1}, T_{k-1})$ and $(S_k, T_k)$ is characterized in the following two lemmas:
\begin{Lem}\label{lem: sink}
If $i_k \in T_{k-1}$, then $(S_k, T_k) = (S_{k-1}, T_{k-1})$.
\end{Lem}
\begin{Lem}\label{lem: source}
If $i_k \in S_{k-1}$, then all the nodes that change their status from $s$ in $G_{k-1}$ to $t$ in $G_k$ must form a (possibly empty) interval of $s$-nodes containing $i_k$ in $G_{k-1}$.
\end{Lem}
Lemma \ref{lem: source} shows that the minimum cut in $G_k$ is derived by updating the minimum cut in $G_{k-1}$ on an interval of nodes that change their status from $s$ to $t$. Note that all nodes that in the status changing interval in Lemma \ref{lem: source} have optimal value $a_{i_k, j_k}$ in GIMR (\ref{prob: GIMR}).

Both Lemma \ref{lem: sink} and \ref{lem: source} hold for any values of $d_{i,i+1}$ and $d_{i+1,i}$ in GIMR (\ref{prob: GIMR}), thus is also true for any values of $\lambda$ in PL-FL (\ref{prob: PL-FL-lambda}). Yet, for different values of $d_{i,i+1}$ and $d_{i+1,i}$ in GIMR (\ref{prob: GIMR}), the node status changing interval in Lemma \ref{lem: source} may be different. This is the place where the values of $d_{i,i+1}$ and $d_{i+1,i}$ in GIMR (\ref{prob: GIMR}), and thus the value of $\lambda$ in PL-FL (\ref{prob: PL-FL-lambda}), affect the optimal solution to GIMR (\ref{prob: GIMR}) and PL-FL (\ref{prob: PL-FL-lambda}) respectively.

\subsubsection{Structure of the path of solutions}\label{sect: path_structure}
According to HL-algorithm for GIMR (\ref{prob: GIMR}), we immediately have the following lemma on the structure of the path of solutions for each node $i$ in PL-FL (\ref{prob: PL-FL-lambda}):
\begin{Lem}\label{lem: piecewise_constant}
For each node $i$, the path of solutions of $x^*_{i}(\lambda)$ for all $\lambda \geq 0$ is piecewise constant. All the constants are taken from the set of breakpoints $\{a_{i_1,j_1}, a_{i_2, j_2}, \ldots, a_{i_q, j_q}\}$.
\end{Lem}

Lemma \ref{lem: piecewise_constant} leads to the following notations and concepts. A \emph{$\lambda$-interval} $\Lambda = [\lambda_{\ell}, \lambda_r]$ is specified by its two endpoints, $\lambda_\ell$ and $\lambda_r$. If $\lambda_\ell = \lambda_r$, the interval $\Lambda$ contains a single value. Let $\Lambda = \emptyset$ if $\lambda_\ell > \lambda_r$. For two disjoint $\lambda$-intervals $\Lambda = [\lambda_\ell, \lambda_r]$ and $\Lambda' = [{\lambda'}_\ell, {\lambda'}_r]$, we define that $\Lambda < \Lambda'$ if $\lambda_r < {\lambda'}_\ell$, and $\Lambda > \Lambda'$ if $\Lambda' < \Lambda$. We say two disjoint $\lambda$-intervals $\Lambda$ and $\Lambda'$ are \emph{adjacent} if $\lambda_r = {\lambda'}_\ell-1$ ($\Lambda < \Lambda'$) or ${\lambda'}_r = \lambda_\ell - 1$ ($\Lambda > \Lambda'$).

We define a $\lambda$-interval $\Lambda(i) = [\lambda_\ell(i), \lambda_r(i)]$ to be a \emph{$\lambda$-constant-interval} for node $i$ if $x^*_{i}(\lambda)$ is a constant for $\lambda \in \Lambda$. And we say a $\lambda$-constant-interval is \emph{maximal} if it is not strictly contained in a larger $\lambda$-interval in which $x^*_{i}(\lambda)$ remains constant. If $\Lambda(i) = [\lambda_\ell(i), \lambda_r(i)]$ is a maximal $\lambda$-constant-interval for node $i$, we call $\lambda_\ell(i)$ (or ($\lambda_\ell(i) - 1$)) is a \emph{$\lambda$-breakpoint} for node $i$, and $\lambda_r(i)$ (or ($\lambda_r(i) + 1$)) is another \emph{$\lambda$-breakpoint} for node $i$. Note that fusing $\lambda$ values are $\lambda$-breakpoints.

Recall that at every fusing $\lambda$ value, some pairs/sets of variables start to always have a same optimal value. Thus we can fuse those variables in PL-FL (\ref{prob: PL-FL-lambda}) to reduce the problem size. Suppose there are $p$ fusing $\lambda$ values, where $p \leq n - 1$ according to Corollary \ref{cor: fusing-num}. Let $\lambda^{(f)}_0 = 0$ and the $j$th ($j\in [p]$) fusing $\lambda$ value\footnote{It could be the case that the value $\lambda = 0$ is already a fusing $\lambda$ value, \ie, the minimizers for $f_i(x_i)$ and $f_{i+1}(x_{i+1})$ are the same for some $i$.} be $\lambda^{(f)}_j$. For each $\lambda^{(f)}_j$, we define a reduced PL-FL problem, namely \emph{PL-FL-$\lambda^{(f)}_j$}, from PL-FL (\ref{prob: PL-FL-lambda}) as follows. For each group of nodes $i_\ell, i_\ell+1, \ldots, i_r-1, i_r$ that are always of a same optimal value for all $\lambda \geq \lambda^{(f)}_j$, we fuse those nodes to generate a \emph{super-node}\footnote{The notation of super-node $I$ in the presentation plays two roles, on one hand it acts as an integer index for the super-node in PL-FL-$\lambda^{f}_j$, on the other hand it refers to the interval $[i_\ell, i_r]$ in PL-FL that the super-node is merged from.} $I_{[i_\ell,i_r]}$ and introduce a new decision variable $x_{I_{[i_\ell, i_r]}}$ in replace of $x_{i_\ell}, x_{i_\ell}+1, \ldots, x_{i_r-1}, x_{i_r}$. We define $f^{pl}_{I_{[i_\ell, i_r]}}(x_{I_{[i_\ell,i_r]}}) = \sum_{i\in I}f^{pl}_i(x_{I_{[i_\ell, i_r]}})$ to replace the original loss functions $\sum_{i\in I}f^{pl}_i(x_i)$. Note that function $f^{pl}_{I_{[i_\ell, i_r]}}$ is also convex piecewise linear, and its breakpoints are union of the breakpoints of functions $f^{pl}_{i_\ell}$ to $f^{pl}_{i_r}$. On the other hand, the term $\lambda|x_{i_\ell-1} - x_{i_\ell}| + \lambda\sum_{i = i_{\ell}}^{i_r-1}|x_i - x_{i+1}| + \lambda|x_{i_r} - x_{i_r+1}|$ is replaced by $\lambda|x_{i_\ell-1} - x_{I_{[i_\ell, i_r]}}| + \lambda|x_{I_{[i_\ell, i_r]}} - x_{i_r+1}|$.

In the following presentation, nodes in PL-FL-$\lambda^{(f)}_j$ are all called \emph{super-nodes} (it could be that a super-node corresponds to a singleton interval, \ie, no fusing) while the term \emph{node} is reserved to nodes in the original PL-FL (\ref{prob: PL-FL-lambda}). We re-index the super-nodes in PL-FL-$\lambda^{(f)}_j$ from $1$ to $n_j$, where $n \geq n_0 > n_1 > \ldots > n_p$. The mappings between the re-indexed values in PL-FL-$\lambda^{(f)}_j$ and the corresponding intervals (could be singleton) in PL-FL is maintained. We define $I_{j,k}$ as the super-node in PL-FL-$\lambda^{(f)}_j$ that contains the node $i_k$ in PL-FL.

All problems, \{PL-FL-$\lambda^{(f)}_0$, PL-FL-$\lambda^{(f)}_1$, \ldots, PL-FL-$\lambda^{(f)}_p$\}, share the same set of piecewise linear breakpoints, $a_{i_1,j_1} < a_{i_2,j_2} < \ldots < a_{i_q,j_q}$, where $a_{i_k,j_k}$ still refers to the breakpoint between the $(j_k-1)$th and the $j_k$th linear pieces of function $f^{pl}_{i_k}(x_{i_k})$ in PL-FL. But note that as the sets of fused nodes are different, the slopes of the linear pieces of the ``fused" loss functions are different among the problems.

PL-FL-$\lambda^{(f)}_j$ has the same optimal solution as PL-FL for all $\lambda \geq \lambda^{(f)}_j$. $x^*_{I_{[i_\ell, i_r]}} = a_{i_k, j_k}$ implies that $x^*_i = a_{i_k, j_k}$ for all $i\in [i_\ell, i_r]$. The PL-FL-$\lambda^{(f)}_j$ problem has exactly the same form as the original PL-FL problem (\ref{prob: PL-FL-lambda}), sharing the same set of breakpoints, yet has smaller number of decision variables. Lemma \ref{lem: piecewise_constant} also applies to PL-FL-$\lambda^{(f)}_j$, so are the concepts of maximal $\lambda$-constant-interval and $\lambda$-breakpoint.


As each PL-FL-$\lambda^{(f)}_j$ is an instance of PL-FL of smaller size, all the prior analysis on PL-FL applies to PL-FL-$\lambda^{(f)}_j$ defined on a bi-path graph of super-nodes. Thus to solve PL-FL-$\lambda^{(f)}_j$, we construct the associated parametric graph $G^{st}_j(\alpha)$ defined over the super-nodes similar to $G^{st}(\alpha)$ for PL-FL. Similarly, we define $G_{j,0} = G^{st}_j(a_{i_1,j_1} - \epsilon)$ for a small value of $\epsilon > 0$ and $G_{j,k} = G^{st}_{j}(a_{i_k, j_k})$ for $k = 1,\ldots, q$. Let $(S_{j,k}, T_{j,k})$ be the minimum cut in $G_{j,k}$ for $k \geq 0$. Based on the HL-algorithm, in PL-FL-$\lambda^{(f)}_j$, we have that for every super-node $I = 1,\ldots,n_j$, $x^*_I = a_{i_k,j_k}$ for some index $k$ such that $I \in S_{j,k-1}$ and $I \in T_{j,k}$. For $\lambda \geq \lambda^{(f)}_j$, it also implies that for any node $i \in I$, $x^*_i = a_{i_k,j_k}$ in PL-FL.


In PL-FL-$\lambda^{(f)}_j$, for any $\lambda \in [\lambda^{(f)}_j, \lambda^{(f)}_{j+1}-1]$ (define $\lambda^{(f)}_{p+1} = +\infty$), no two adjacent super-nodes take a same optimal value because there is no value fusion for $\lambda \in [\lambda^{(f)}_j, \lambda^{(f)}_{j+1}-1]$. On the other hand, from HL-algorithm in \cite{HL17} for PL-FL (\ref{prob: PL-FL-lambda}), we observe that any two adjacent nodes are of the same optimal value, say $a_{i_k,j_k}$, only if they are both in the source set in $G_{k-1}$ but shift to the sink set in $G_k$. Furthermore, at $G_k$, if there is at least one node shifted to the sink set, node $i_k$ must be one of them. Combining the above three observations, we have the following key insight on PL-FL-$\lambda^{(f)}_j$:
\begin{Lem}\label{lem: one-node-shift}
For any $\lambda \in [\lambda^{(f)}_j, \lambda^{(f)}_{j+1}-1]$, if there exits at least one super-node that is in $S_{j,k-1}$ in $G_{j,k-1}$ but shifts to $T_{j,k}$ in $G_{j,k}$, it must be the super-node $I_{j,k}$ that contains $i_k$ in PL-FL. 
\end{Lem}

Based on the above observation, we shall prove the following theorem bounding the number of different $\lambda$-breakpoints of \emph{all} super-nodes in PL-FL-$\lambda^{(f)}_j$ for $\lambda \in [\lambda^{(f)}_j, \lambda^{(f)}_{j+1}-1]$:
\begin{Thm}\label{thm: induced_graph_breakpoints}
In PL-FL-$\lambda^{(f)}_j$, the number of $\lambda$-breakpoints of all super-nodes for $\lambda\in [\lambda^{(f)}_j, \lambda^{(f)}_{j+1}-1]$ is at most $q$.
\end{Thm}

The remainder of the section is to prove Theorem \ref{thm: induced_graph_breakpoints}. To do so, we define some additional concepts. In PL-FL-$\lambda^{(f)}_j$, for every super-node $I$ in the parametric graph $G^{st}_j(\alpha)$, if $I$ is in the source set for $\lambda$-interval $[\lambda_\ell, \lambda_r]$, we call $[\lambda_\ell, \lambda_r]$ an \emph{$I$-source-$\lambda$-interval}; if $I$ is in the sink set for $\lambda$-interval $[\lambda_\ell, \lambda_r]$, we call $[\lambda_\ell, \lambda_r]$ an \emph{$I$-sink-$\lambda$-interval}.

Initially in $G_{j,0}$, as only source adjacent arcs have non-zero capacities, thus all super-nodes in $G_{j,0}$ are in $S_{j,0}$. Hence $[\lambda^{(f)}_j, \lambda^{(f)}_{j+1}-1]$ is $I$-source-$\lambda$-interval for every super-node $I$ in $G_{j,0}$. At $G_{j,k}$, some subintervals of $[\lambda^{(f)}_j, \lambda^{(f)}_{j+1}-1]$ may change from $I_{j,k}$-source-$\lambda$-intervals to $I_{j,k}$-sink-$\lambda$-intervals. To compute those subintervals, we solve the values of $\lambda$ such that the single super-node $I_{j,k}$ shifts from $S_{j,k-1}$ in $G_{j,k-1}$ to $T_{j,k}$ in $G_{j,k}$. On the other hand, for any $\lambda \in [\lambda^{(f)}_j, \lambda^{(f)}_{j+1}-1]$, all nodes except $I_{j,k}$ must have the same status between $G_{j,k-1}$ and $G_{j,k}$. As a result, one can easily solve the status of $I_{j,k}$ in $G_{j,k}$ for different values of $\lambda$, depending on the status of the two adjacent super-nodes of $I_{j,k}$, $I_{j,k}-1$ and $I_{j,k}+1$ (if exist):
\begin{Prop}\label{prop: lambda_ranges}
In $G_{j,k}$ for PL-FL-$\lambda^{(f)}_j$, if $I_{j,k} \in S_{j,k-1}$, then the status of $I_{j,k}$ is determined by the status of its two adjacent super-nodes $I_{j,k}-1, I_{j,k}+1$, and the value of $\lambda$, in the following way:
\begin{enumerate}
\item If both $I_{j,k}-1$ and $I_{j,k}+1$ exist ($1 < I_{j,k} < n_j$):
    \begin{enumerate}
    \item $I_{j,k}-1, I_{j,k}+1 \in S_{j,k-1}$: For $\lambda \in [0, (c_{I_{j,k},t} - c_{s,I_{j,k}})/2)$, $I_{j,k} \in T_{j,k}$; otherwise $I_{j,k} \in S_{j,k}$. Note that if $c_{I_{j,k},t} - c_{s,I_{j,k}} \leq 0$, then the interval $[0, (c_{I_{j,k},t} - c_{s,I_{j,k}}/2)$ is empty, thus $I_{j,k} \in S_{j,k}$ for all $\lambda \in [\lambda^{(f)}_j, \lambda^{(f)}_{j+1}-1]$.
    \item $I_{j,k}-1 \in S_{j,k-1}, I_{j,k}+1 \in T_{j,k-1}$, or the reverse: If $c_{I_{j,k},t} - c_{s,I_{j,k}} > 0$, then $I_{j,k} \in T_{j,k}$ for all $\lambda \in [\lambda^{(f)}_j, \lambda^{(f)}_{j+1}-1]$; otherwise $I_{j,k} \in S_{j,k}$ for all $\lambda \in [\lambda^{(f)}_j, \lambda^{(f)}_{j+1}-1]$.
    \item $I_{j,k}-1, I_{j,k}+1 \in T_{j,k-1}$: For $\lambda \in ((c_{s,I_{j,k}} - c_{I_{j,k},t})/2, +\infty)$, $I_{j,k} \in T_{j,k}$; otherwise $I_{j,k} \in S_{j,k}$. Note that if $c_{s,I_{j,k}} - c_{I_{j,k},t} < 0$, then for all $\lambda \in [\lambda^{(f)}_j, \lambda^{(f)}_{j+1}-1]$, $I_{j,k} \in T_{j,k}$, as $\lambda \geq 0$.
    \end{enumerate}
\item If either $I_{j,k}-1$ doesn't exist ($I_{j,k}=1$) or $I_{j,k}+1$ doesn't exist ($I_{j,k} = n_j$): W.l.o.g., we consider the case where $I_{j,k} = 1$, thus $I_{j,k}-1$ doesn't exist.
    \begin{enumerate}
    \item $I_{j,k}+1 \in S_{j,k-1}$: For $\lambda \in [0, c_{I_{j,k},t} - c_{s,I_{j,k}})$, $I_{j,k}\in T_{j,k}$; otherwise $I_{j,k}\in S_{j,k}$. Note that if $c_{I_{j,k},t} - c_{s,I_{j,k}} \leq 0$, then the interval $[0, c_{I_{j,k},t} - c_{s,I_{j,k}})$ is empty, thus $I_{j,k} \in S_{j,k}$ for all $\lambda\in [\lambda^{(f)}_j, \lambda^{(f)}_{j+1}-1]$.
    \item $I_{j,k}+1 \in T_{j,k-1}$: For $\lambda \in ((c_{s,I_{j,k}} - c_{I_{j,k},t}), +\infty)$, $I_{j,k} \in T_{j,k}$; otherwise $I_{j,k} \in S_{j,k}$. Note that if $c_{s,I_{j,k}} - c_{I_{j,k},t} < 0$, then for all $\lambda \in [\lambda^{(f)}_j, \lambda^{(f)}_{j+1}-1]$, $I_{j,k} \in T_{j,k}$, as $\lambda \geq 0$.
    \end{enumerate}
\end{enumerate}
\end{Prop}
\begin{proof}
The proof is by straightforward computation and comparison. We only show the case 1-(a). The other cases can be derived similarly. If $I_{j,k} \in S_{j,k}$ in $G_{j,k}$, then $S_{j,k} = S_{j,k-1}$. Thus the cut capacity in $G_{j,k}$ is
\begin{equation*}
\begin{split}
C_1 &= C(\{s\}\cup S_{j,k-1}, T_{j,k-1}\cup\{t\}) \\
&= C\big(\{s\}\cup (S_{j,k-1}\cap([1,n_j]\setminus\{I_{j,k}\})), (T_{j,k}\cap([1,n_j]\setminus\{I_{j,k}\}))\cup \{t\}\big) + c_{I_{j,k},t}.
\end{split}
\end{equation*}
If $I_{j,k}\in T_{j,k}$ in $G_{j,k}$, then $S_{j,k} = S_{j,k-1}\setminus \{I_{j,k}\}$ and $T_{j,k} = T_{j,k-1}\cup \{I_{j,k}\}$. Thus the cut capacity in $G_{j,k}$ is
\begin{equation*}
\begin{split}
C_2 &= C(\{s\}\cup (S_{j,k-1}\setminus\{I_{j,k}\}), (T_{j,k-1}\cup\{I_{j,k}\})\cup\{t\}) \\
&= C\big(\{s\}\cup (S_{j,k-1}\cap([1,n_j]\setminus\{I_{j,k}\})), (T_{j,k}\cap([1,n_j]\setminus\{I_{j,k}\}))\cup \{t\}\big) + c_{s,I_{j,k}} + 2\lambda.
\end{split}
\end{equation*}
If $C_1 \leq C_2$, \ie, $\lambda \geq (c_{I_{j,k},t} - c_{s,I_{j,k}})/2$, then $I_{j,k} \in S_{j,k}$ (recall that we always select the maximal source set), otherwise $I_{j,k}\in T_{j,k}$.
\end{proof}

We observe from Proposition \ref{prop: lambda_ranges} that the two terms, $c_{s,I_{j,k}} - c_{I_{j,k},t}$ and $c_{I_{j,k},t} - c_{s,I_{j,k}}$, play important roles in determining the ranges of $\lambda$. Thus for ease of presentation, we define two shortcut terms $smt_{j,k}(I) = c_{s,I} - c_{I,t}$ and $tms_{j,k}(I) = -smt_{j,k}(I) = c_{I,t} - c_{s,I}$ for each super-node $I$ in $G_{j,k}$ (the $m$ in the notation refers to ``minus"). Note that for any super-node $I$, $smt_{j,k}(I)$ is nonincreasing in $k$, correspondingly $tms_{j,k}(I)$ is nondecreasing in $k$. Also note that, for a fixed $k$ and $I$, the $smt$ and $tms$ values are different in $j$ (for different reduced PL-FL problems, PL-FL-$\lambda^{(f)}_j$). This is important as the two quantities determine the $\lambda$-breakpoints.

As we care only the case where $I_{j,k}$ shifts from $S_{j,k-1}$ to $T_{j,k}$, we focus on the conditions under which the source to sink shift happens. From Proposition \ref{prop: lambda_ranges}, we observe that, if $smt_{j,k}(I_{j,k}) < 0$ ($tms_{j,k}(I_{j,k}) > 0$), no matter what the status of its adjacent super-nodes are in, there could exist a range of $\lambda$ in $[\lambda^{(f)}_j, \lambda^{(f)}_{j+1} - 1]$ that $I_{j,k}$ may shift to the sink set. If $smt_{j,k}(I_{j,k}) \geq 0$ ($tms_{j,k}(I_{j,k}) \leq 0$), however, there is a restriction on the status of $I_{j,k}$'s adjacent super-nodes, as follows:
\begin{Cor}\label{cor: smt_condition}
Based on Proposition \ref{prop: lambda_ranges}, if If $smt_{j,k}(I_{j,k}) \geq 0$ ($tms_{j,k}(I_{j,k}) \leq 0$), a necessary condition for $I_{j,k}$ to shift from $S_{j,k-1}$ in $G_{j,k-1}$ to $T_{j,k}$ in $G_{j,k}$ is that both the adjacent super-nodes, $I_{j,k} - 1$ and $I_{j,k} + 1$ (if exist), must be in the sink set $T_{j,k-1}$ (so is in $T_{j,k}$).
\end{Cor}

The following lemma is key to prove Theorem \ref{thm: induced_graph_breakpoints}:
\begin{Lem}\label{lem: m_breakpoints}
After the computation of minimum cut $(S_{j,k}, T_{j,k})$ for $G_{j,k}$, for each super-node $I$ in $G_{j,k}$,
\begin{enumerate}
\item If $smt_{j,k}(I) \geq 0$,
  there exists a subinterval, possibly empty, $\emptyset \subseteq [\lambda_{k,\ell}(I), \lambda_{k,r}(I)] \subseteq [\lambda^{(f)}_j, \lambda^{(f)}_{j+1}-1]$ such that $[\lambda_{k,\ell}(I), \lambda_{k,r}(I)]$ is an $I$-sink-$\lambda$-interval, $[\lambda^{(f)}_j, \lambda_{k,\ell}(I)-1]$ and $[\lambda_{k,r}(I)+1, \lambda^{(f)}_{j+1}-1]$ are both $I$-source-$\lambda$-intervals. Note that if $[\lambda_{k,\ell}(I), \lambda_{k,r}(I)] = \emptyset$, then the whole interval $[\lambda^{(f)}_j, \lambda^{(f)}_{j+1}-1]$ is an $I$-source-$\lambda$-interval.
\item If $smt_{j,k}(I) < 0$, there exists a $\lambda_k(I) \in [\lambda^{(f)}_j-1, \lambda^{(f)}_{j+1}-1]$ such that $[\lambda^{(f)}_j, \lambda_k(I)]$ is an $I$-sink-$\lambda$-interval and $[\lambda_k(I)+1, \lambda^{(f)}_{j+1}-1]$ is an $I$-source-$\lambda$-interval.
\end{enumerate}
\end{Lem}
\begin{proof}
We prove the result by induction on $k$, $k = 0,1,\ldots, q$.

The lemma holds for $k = 0$ because in $G_{j,0}$, $[\lambda^{(f)}_j, \lambda^{(f)}_{j+1}-1]$ is an $I$-source-$\lambda$-interval for all $I \in [1,\ldots,n_j]$. Thus for each super-node $I$, since $smt_{j,0}(I) \geq 0$, we have $[\lambda_{0,\ell}(I), \lambda_{0,r}(I)] = \emptyset$.

Suppose the lemma holds for $k-1 \geq 0$. We prove the lemma also holds for $k$. In $G_{j,k}$, as the only node that can possibly change status is $I_{j,k}$, we only need to consider the node $I_{j,k}$. For other super-node $I \neq I_{j,k}$, we have $[\lambda_{k,\ell}(I), \lambda_{k,r}(I)] = [\lambda_{k-1,\ell}(I), \lambda_{k-1,r}(I)]$ if $smt_{j,k}(I) = smt_{j,k-1}(I) \geq 0$, or $\lambda_{k}(I) = \lambda_{k-1}(I)$ if $smt_{j,k}(I) = smt_{j,k-1}(I) < 0$. We only need prove that, if there is some $I_{j,k}$-source-$\lambda$-interval changes to an $I_{j,k}$-sink-$\lambda$-interval, the lemma still holds for $I_{j,k}$.


Depending on the sign of $smt_{k-1}(I_{j,k})$ in $G_{j,k-1}$, we consider the two cases separately:
\begin{enumerate}
\item $smt_{j,k-1}(I_{j,k}) \geq 0$ in $G_{j,k-1}$: We first show that, after the computation of minimum cut in $G_{j,k-1}$, for super-nodes $I_{j,k} - 1$ and $I_{j,k} + 1$, either of the following two cases must hold:
	\begin{enumerate}
	\item $[\lambda^{(f)}_j, \lambda^{(f)}_{j+1} - 1]$ is $(I_{j,k}-1)$-source-$\lambda$-interval and $(I_{j,k} + 1)$-source-$\lambda$-interval in $G_{j,k-1}$, so is $G_{j,k}$.
	\item $smt_{j,k-1}(I_{j,k}-1) < 0$ and $smt_{j,k-1}(I_{j,k}+1) < 0$ in $G_{j,k-1}$, so is $G_{j,k}$.
	\end{enumerate}
	Suppose case (a) does not hold, it implies that there are $\lambda$-intervals in $[\lambda^{(f)}_j, \lambda^{(f)}_{j+1} - 1]$ for which $I_{j,k}-1, I_{j,k}+1 \in T_{j,k-1}$. According to Corollary \ref{cor: smt_condition}, since $smt_{p}(I_{j,k}) \geq 0$ for $0 \leq p \leq k-1$, if there ever was a $\lambda$-interval for $I_{j,k}$ that changed from $I_{j,k}$-source-$\lambda$-interval to $I_{j,k}$-sink-$\lambda$-interval, say in $G_{j,k'}\ (k' \leq k-1)$, then both $I_{j,k} - 1$ and $I_{j,k} + 1$ must have been in the sink set. In other words, there exists an $k'' < k'$ such that in $G_{j,k''}$, some $(I_{j,k}-1)$- and $(I_{j,k} + 1)$-sink-$\lambda$-intervals were generated. As in $G_{j,k''}$, $I_{j,k} \in S_{j,k''}$ for $\lambda \in [\lambda^{(f)}_j, \lambda^{(f)}_{j+1}-1]$, according to Proposition \ref{prop: lambda_ranges}, we have $smt_{k''}(I_{j,k}-1) < 0$ and $smt_{k''}(I_{j,k} + 1) < 0$ in $G_{j,k''}$. Recall that $smt_{j,k}(I)$ is nonincreasing in $k$, and $k'' < k' \leq k-1$, therefore case (b) holds.

By the above derivation, if case (a) holds for both $I_{j,k} - 1$ and $I_{j,k} + 1$, then $[\lambda^{(f)}_j, \lambda^{(f)}_{j+1} - 1]$ is also an $I_{j,k}$-source-$\lambda$-interval in $G_{j,k-1}$.

If case (b) holds, then by induction hypothesis, there exist $\lambda_{k-1}(I_{j,k}-1)$ and $\lambda_{k-1}(I_{j,k}+1)$ such that case 2 of the lemma holds.

In $G_{j,k}$, it could be either case that $smt_{j,k}(I_{j,k}) \geq 0$ or $smt_{j,k}(I_{j,k}) < 0$. We consider the two sub-cases separately:
\begin{enumerate}
\item[i.] $smt_{j,k}(I_{j,k}) \geq 0$ in $G_{j,k}$:

If case (a) holds, then according to Proposition \ref{prop: lambda_ranges}, $[\lambda^{(f)}_j, \lambda^{(f)}_{j+1}-1]$ remains $I_{j,k}$-source-$\lambda$-interval in $G_{j,k}$. The lemma holds with $[\lambda_{k,\ell}(I_{j,k}), \lambda_{k,r}(I_{j,k})] = \emptyset$.



    If case (b) holds, then by Proposition \ref{prop: lambda_ranges}, in $G_{j,k}$, only for $\lambda \in [\lceil smt_{j,k}(I_{j,k})/2 \rceil, \min\{\lambda_{k}(I_{j,k}-1), \lambda_{k}(I_{j,k}+1)\}]$, if $I_{j,k} \in S_{j,k-1}$ in $G_{j,k-1}$, then $I_{j,k}$ shifts to $T_{j,k}$ in $G_{j,k}$. According to the induction hypothesis, if the interval $[\lambda_{k-1,\ell}(I_{j,k}), \lambda_{k-1,r}(I_{j,k})]$ is empty, then in $G_{j,k}$, we have
    \begin{equation}\label{eqn: lambda-breakpoint1}
    \begin{split}
    &\lambda_{k,\ell}(I_{j,k}) = \max\{\lceil smt_{j,k}(I_{j,k})/2 \rceil, \lambda^{(f)}_j\},\\
    &\lambda_{k,r}(I_{j,k}) = \min\{\lambda_{k}(I_{j,k}-1), \lambda_{k}(I_{j,k}+1)\}.
    \end{split}
    \end{equation}
     Otherwise, there exists a $k' < k$ such that $\lambda_{k-1,\ell}(I_{j,k}) = \lambda_{k',\ell}(I_{j,k}) = \max\{\lceil smt_{k'}(I_{j,k}) / 2 \rceil, \lambda^{(f)}_j\}$ and $\lambda_{k-1,r}(I_{j,k}) = \lambda_{k',r}(I_{j,k}) = \min\{\lambda_{k'}(I_{j,k}-1), \lambda_{k'}(I_{j,k}+1)\}$. Since $smt_{j,k}(I_{j,k}) \leq smt_{k'}(I_{j,k})$, $\lambda_{k}(I_{j,k} - 1) \geq \lambda_{k'}(I_{j,k}-1)$, and $\lambda_{k}(I_{j,k} + 1) \geq \lambda_{k'}(I_{j,k}+1)$, we have, in $G_{j,k}$,
    \begin{equation*}
    [\lambda_{k-1, \ell}(I_{j,k}), \lambda_{k-1,r}(I_{j,k})] \subseteq [\max\{\lceil smt_{j,k}(I_{j,k})/2 \rceil, \lambda^{(f)}_j\}, \min\{\lambda_{k}(I_{j,k}-1), \lambda_{k}(I_{j,k}+1)\}].
    \end{equation*}
Hence in $G_{j,k}$,
   \begin{equation}\label{eqn: lambda-breakpoint2}
   [\lambda_{k,\ell}(I_{j,k}), \lambda_{k,r}(I_{j,k})] = [\max\{\lceil smt_{j,k}(I_{j,k})/2 \rceil, \lambda^{(f)}_j\}, \min\{\lambda_{k}(I_{j,k}-1), \lambda_{k}(I_{j,k}+1)\}].
   \end{equation}
   The lemma holds. The case (\ref{eqn: lambda-breakpoint2}) is illustrated in Figure \ref{figure: smt_positive_positive}. Note that at most one additional $\lambda$-breakpoint, $\lceil smt_{j,k}(I_{j,k})/2 \rceil$, is increased.
    \begin{figure}[!h]
    \centering
    \includegraphics[scale=0.8]{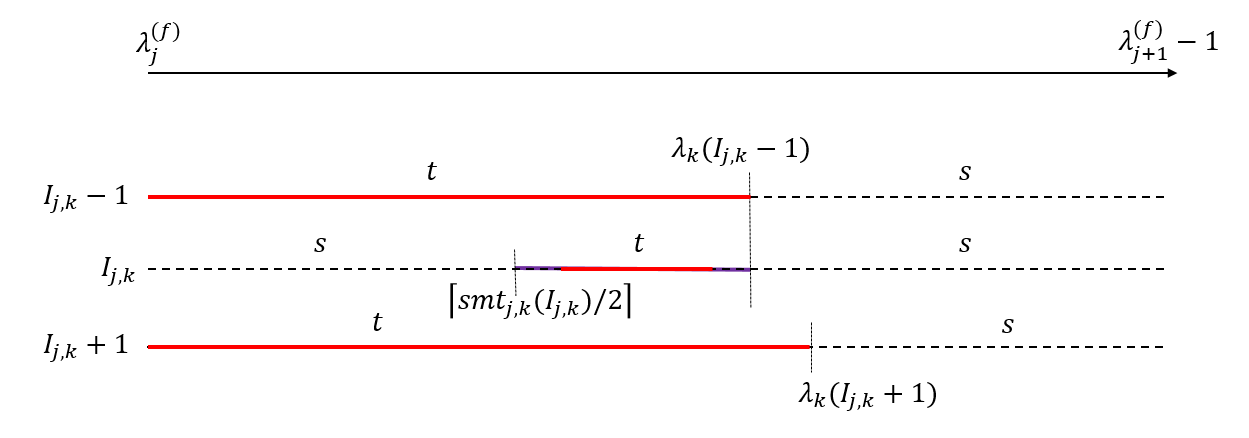}
    \caption{Illustration of the case (\ref{eqn: lambda-breakpoint2}). In this case, $smt_{j,k-1}(I_{j,k}) \geq 0$ in $G_{j,k-1}$, $smt_{j,k}(I_{j,k}) \geq 0$ in $G_{j,k}$, and $smt_{j,k}(I_{j,k}-1), smt_{j,k}(I_{j,k}+1) < 0$. The top black arrow denotes the $\lambda$ interval $[\lambda^{(f)}_j, \lambda^{(f)}_{j+1}-1]$. For super-nodes $I_{j,k}-1$, $I_{j,k}$ and $I_{j,k}+1$, the respective solid red line segment denotes the $\lambda$ interval for which the super-node is in $T_{j,k-1}$ in $G_{j,k-1}$. The solid purple line segments denote the new $I_{j,k}$-sink-$\lambda$-intervals introduced in $G_{j,k}$. The remaining dashed line segments denote the source-$\lambda$-intervals in $G_{j,k}$.}\label{figure: smt_positive_positive}
    \end{figure}

\item[ii.] $smt_{j,k}(I_{j,k}) < 0$ in $G_{j,k}$:

If case (a) holds, then according to Proposition \ref{prop: lambda_ranges}, if $\lfloor tms_{j,k}(I_{j,k}) / 2 \rfloor > \lambda^{(f)}_j$, then
\begin{equation}\label{eqn: lambda-breakpoint3}
\begin{split}
\lambda_{k}(I_{j,k}) = \min\{\lfloor tms_{j,k}(I_{j,k})/2 \rfloor, \lambda^{(f)}_{j+1}-1\}
\end{split}
\end{equation}
such that $[\lambda^{(f)}_j, \lambda_{k}(I_{j,k})]$ becomes $I_{j,k}$-sink-$\lambda$-interval in $G_{j,k}$, otherwise $[\lambda^{(f)}_j, \lambda^{(f)}_{j+1}-1]$ remains an $I_{j,k}$-source-$\lambda$-interval ($\lambda_{k}(I_{j,k}) = \lambda^{(f)}_j - 1$). Thus the lemma holds. The case (\ref{eqn: lambda-breakpoint3}) is illustrated in Figure \ref{figure: smt_positive_negative}.
    \begin{figure}[!h]
    \centering
    \includegraphics[scale=0.8]{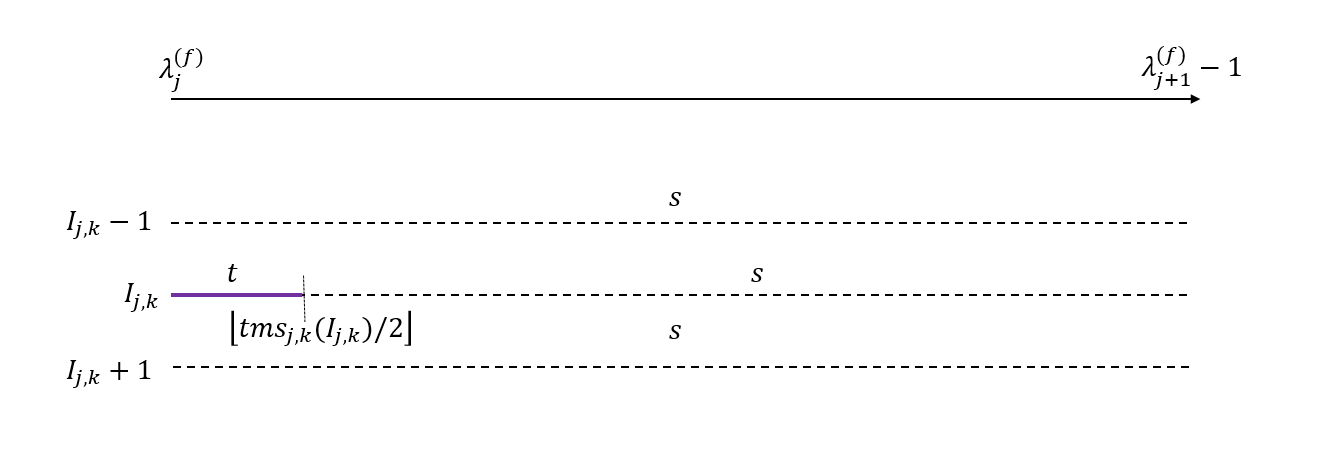}
    \caption{Illustration of the case (\ref{eqn: lambda-breakpoint3}). In this case, $smt_{j,k-1}(I_{j,k}) \geq 0$ in $G_{j,k-1}$, $smt_{j,k}(I_{j,k}) < 0$ in $G_{j,k}$, and $[\lambda^{(f)}_j, \lambda^{(f)}_{j+1}-1]$ is $(I_{j,k}-1)$-source-$\lambda$-interval and $(I_{j,k}+1)$-source-$\lambda$-interval in $G_{j,k}$. The top black arrow denotes the $\lambda$ interval $[\lambda^{(f)}_j, \lambda^{(f)}_{j+1}-1]$. The solid purple line segments denote the new $I_{j,k}$-sink-$\lambda$-intervals introduced in $G_{j,k}$. The remaining dashed line segments denote the source-$\lambda$-intervals in $G_{j,k}$.}\label{figure: smt_positive_negative}
    \end{figure}

If case (b) holds, then by Proposition \ref{prop: lambda_ranges}, the interval $[\lambda^{(f)}_j, \lambda^{(\max)}_{k} = \max\{ \lambda_{k}(I_{j,k}-1), \lambda_{k}(I_{j,k}+1)\}]$ must be $I_{j,k}$-sink-$\lambda$-interval, as at least one of $I_{j,k}-1$ and $I_{j,k} + 1$ is in $T_{j,k-1}$. If $\lfloor tms_{j,k}(I_{j,k}) / 2\rfloor > \lambda^{(\max)}_{k}$, the right endpoint can be further extended to $\min\{\lfloor tms_{j,k}(I_{j,k}) / 2\rfloor, \lambda^{(f)}_{j+1}-1\}$. Therefore we have
	\begin{equation}\label{eqn: lambda-breakpoint4}
	\lambda_k(I_{j,k}) = \begin{cases}
					\min\{\lfloor tms_{j,k}(I_{j,k}) / 2\rfloor, \lambda^{(f)}_{j+1}-1\},\ &\text{if}\ \lfloor tms_{j,k}(I_{j,k}) / 2\rfloor > \lambda^{(\max)}_{k},\\
					\lambda^{(\max)}_{k},\ &\text{otherwise}.
					\end{cases}
	\end{equation}
	The case (\ref{eqn: lambda-breakpoint4}) is illustrated in Figure \ref{figure: smt_positive_negative_2}. Note that in both cases, at most one additional $\lambda$-breakpoint, $\lfloor tms_{j,k}(I_{j,k}) / 2\rfloor$, is introduced.
    \begin{figure}[!h]
    \centering
    \includegraphics[scale=0.8]{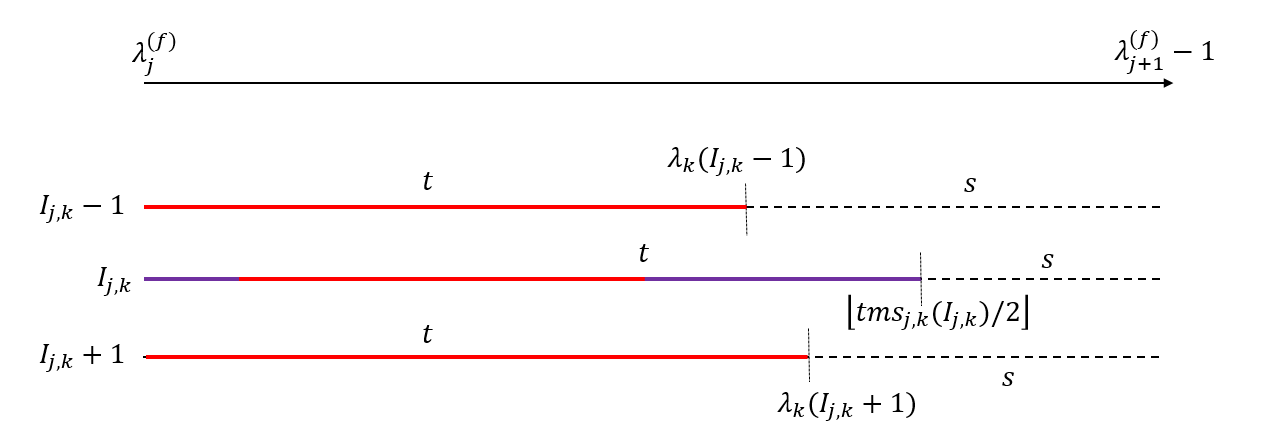}
    \caption{Illustration of the case (\ref{eqn: lambda-breakpoint4}). In this case, $smt_{j,k-1}(I_{j,k}) \geq 0$ in $G_{j,k-1}$, $smt_{j,k}(I_{j,k}) < 0$ in $G_{j,k}$, and $smt_{j,k}(I_{j,k}-1), smt_{j,k}(I_{j,k}+1) < 0$. The top black arrow denotes the $\lambda$ interval $[\lambda^{(f)}_j, \lambda^{(f)}_{j+1}-1]$. For super-nodes $I_{j,k}-1$, $I_{j,k}$ and $I_{j,k}+1$, the respective solid red line segment denotes the $\lambda$ interval for which the super-node is in $T_{j,k-1}$ in $G_{j,k-1}$. The solid purple line segments denote the new $I_{j,k}$-sink-$\lambda$-intervals introduced in $G_{j,k}$. The remaining dashed line segments denote the source-$\lambda$-intervals in $G_{j,k}$.}\label{figure: smt_positive_negative_2}
    \end{figure}
\end{enumerate}

\item $smt_{j,k-1}{I_{j,k}} < 0$ in $G_{j,k-1}$: It implies that $smt_{j,k}(I_{j,k}) \leq smt_{j,k-1}(I_{j,k}) < 0$ in $G_{j,k}$. By induction hypothesis, there exists an $\lambda_{k-1}(I_{j,k})$ such that $[\lambda_{k-1}(I_{j,k})+1, \lambda^{(f)}_{j+1}-1]$ is an $I_{j,k}$-source-$\lambda$-interval in $G_{j,k-1}$. Consider the status of super-nodes $I_{j,k}-1$ and $I_{j,k}+1$. If $smt_{j,k-1}(I_{j,k}-1) \geq 0$, since $I_{j,k}$ is $s$-super-node until $G_{j,k-1}$ for $\lambda \in [\lambda_{k-1}(I_{j,k})+1, \lambda^{(f)}_{j+1}-1]$, then by Proposition \ref{prop: lambda_ranges}, it must be that $[\lambda_{k-1}(I_{j,k})+1, \lambda^{(f)}_{j+1}-1]$ is an $(I_{j,k}-1)$-source-$\lambda$-interval in $G_{j,k-1}$, so is $G_{j,k}$. Otherwise $smt_{j,k-1}(I_{j,k}-1) < 0$, then by the induction hypothesis, there exists a $\lambda_{k-1}(I_{j,k}-1)$ such that $[\lambda^{(f)}_j, \lambda_{k-1}(I_{j,k}-1)]$ is an $(I_{j,k}-1)$-sink-$\lambda$-interval and $[\lambda_{k-1}(I_{j,k}-1)+1, \lambda^{(f)}_{j+1}-1]$ is an $(I_{j,k}-1)$-source-$\lambda$-interval. The same results hold for super-node $I_{j,k} + 1$.

To summarize, in the $I_{j,k}$-source-$\lambda$-interval $[\lambda_{k-1}(I_{j,k})+1, \lambda^{(f)}_{j+1}-1]$, for each of the two super-nodes $I_{j,k} - 1$ and $I_{j,k}+1$, either the the whole interval $[\lambda_{k-1}(I_{j,k})+1, \lambda^{(f)}_{j+1}-1]$ is a source-$\lambda$-interval, or the interval is dichotomized into two segments, with the left being a sink-$\lambda$-interval and the right being a source-$\lambda$-interval.

    In the first case, by Proposition \ref{prop: lambda_ranges}, in $G_{j,k}$, we have
	\begin{equation}\label{eqn: lambda-breakpoint5}
	\lambda_k(I_{j,k}) = \begin{cases}
    				\min\{\lfloor tms_{j,k}(I_{j,k})/2 \rfloor, \lambda^{(f)}_{j+1}-1\},\ &\text{if}\ \lfloor tms_{j,k}(I_{j,k})/2 \rfloor > \lambda_{k-1}(I_{j,k}),\\
				\lambda_{k-1}(I_{j,k}),\ &\text{otherwise}.
				    \end{cases}
	\end{equation}
    One example of the case (\ref{eqn: lambda-breakpoint5}) is illustrated in Figure \ref{figure: smt_negative_1}.
    \begin{figure}[!h]
    \centering
    \includegraphics[scale=0.8]{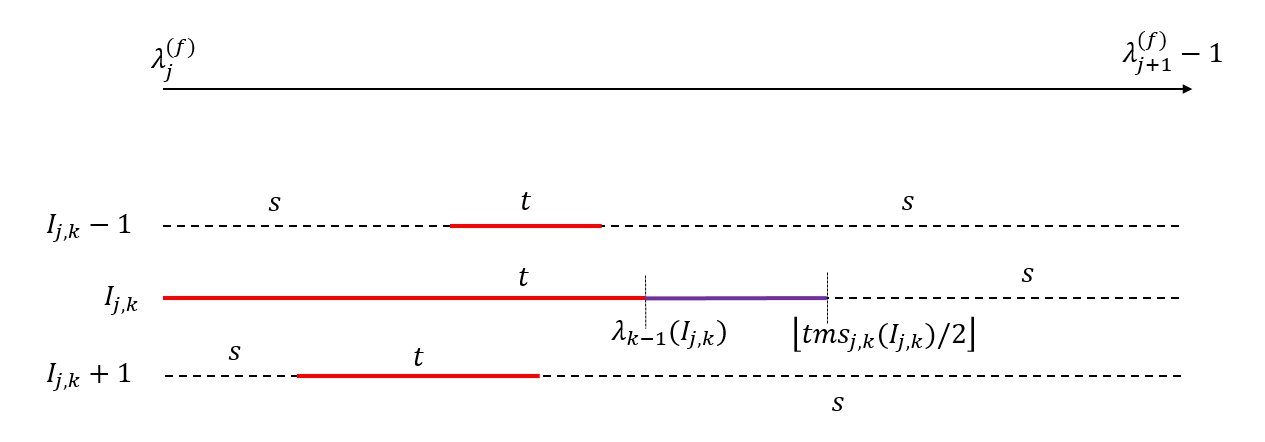}
    \caption{Illustration of one example of the case (\ref{eqn: lambda-breakpoint5}). In this example, for $I_{j,k}$, $smt_{j,k-1}(I_{j,k}) < 0$ in $G_{j,k-1}$ and $smt_{j,k}(I_{j,k}) < 0$ in $G_{j,k}$. For $I_{j,k}-1$ and $I_{j,k} + 1$, $smt_{j,k}(I_{j,k}-1), smt_{j,k}(I_{j,k}+1) \geq 0$ (it could be $smt_{j,k}(I_{j,k}-1) < 0$ and/or $smt_{j,k}(I_{j,k}+1) < 0$ that falls into the first case of (\ref{eqn: lambda-breakpoint5})). The top black arrow denotes the $\lambda$ interval $[\lambda^{(f)}_j, \lambda^{(f)}_{j+1}-1]$. For super-nodes $I_{j,k}-1$, $I_{j,k}$ and $I_{j,k}+1$, the respective solid red line segment denotes the $\lambda$ interval for which the super-node is in $T_{j,k-1}$ in $G_{j,k-1}$. The solid purple line segments denote the new $I_{j,k}$-sink-$\lambda$-intervals introduced in $G_{j,k}$. The remaining dashed line segments denote the source-$\lambda$-intervals in $G_{j,k}$.}\label{figure: smt_negative_1}
    \end{figure}

In the second case, let $\lambda^{(\max)}_{k}$ be the dichotomy point in $[\lambda_{k-1}(I_{j,k})+1, \lambda^{(f)}_{j+1}-1]$ such that in $[\lambda_{k-1}(I_{j,k}) + 1, \lambda^{\max}_{k}]$ is a sink-$\lambda$-interval for at least one of $I_{j,k}-1$ or $I_{j,k}+1$, and $[\lambda^{\max}_{k} + 1, \lambda^{(f)}_{j+1} - 1]$ is source-$\lambda$-interval for both $I_{j,k}-1$ and $I_{j,k}+1$. Note that $\lambda^{(\max)}_{k}$ is the largest of $\lambda_{k}(I_{j,k} - 1)\ (=\lambda_{k-1}(I_{j,k}-1))$ and $\lambda_{k}(I_{j,k} + 1)\ (=\lambda_{k-1}(I_{j,k}+1))$, if exist. Then in $G_{j,k}$, by Proposition \ref{prop: lambda_ranges}, we have
	\begin{equation}\label{eqn: lambda-breakpoint6}
	\lambda_k(I_{j,k}) = \begin{cases}
    				\min\{\lfloor tms_{j,k}(I_{j,k})/2 \rfloor, \lambda^{(f)}_{j+1}-1\},\ &\text{if}\ \lfloor tms_{j,k}(I_{j,k})/2 \rfloor > \lambda^{(\max)}_{k},\\
				\lambda^{(\max)}_{k},\ &\text{otherwise}.
				    \end{cases}
	\end{equation}
    Two examples of the case (\ref{eqn: lambda-breakpoint6}) are illustrated in Figure \ref{figure: smt_negative_2}. Note that in either case, at most one additional $\lambda$-breakpoint, $\lfloor tms_{j,k}(I_{j,k})/2 \rfloor$, is introduced.
    \begin{figure}[!h]
    \centering
    \subfloat[$smt_{j,k}(I_{j,k}-1) \geq 0$ and $smt_{j,k}(I_{j,k}+1) < 0$ in $G_{j,k}.$]{\includegraphics[scale=0.8]{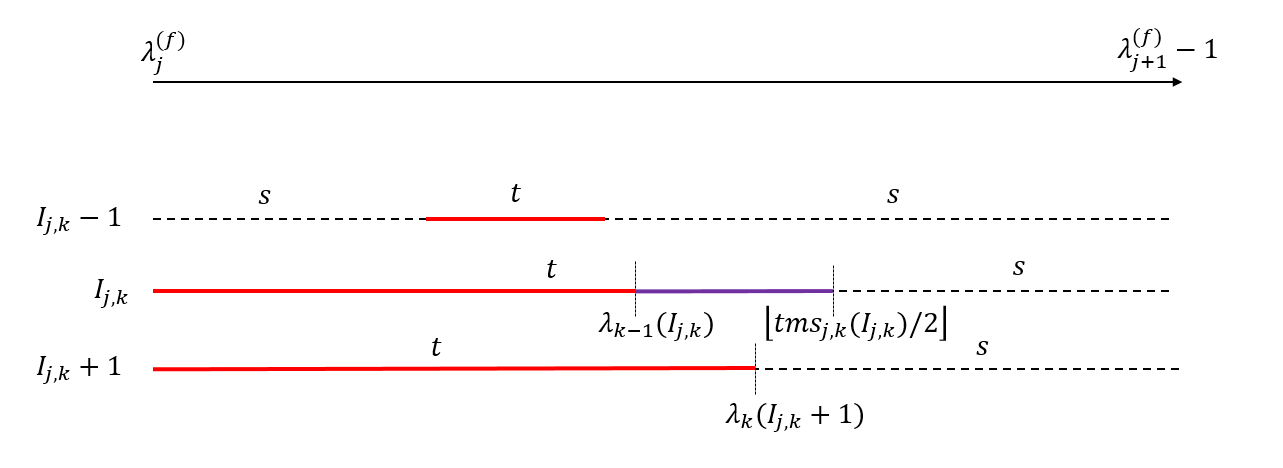}}\\
    \subfloat[$smt_{j,k}(I_{j,k}-1) < 0$ and $smt_{j,k}(I_{j,k}+1) < 0$ in $G_{j,k}$.]{\includegraphics[scale=0.8]{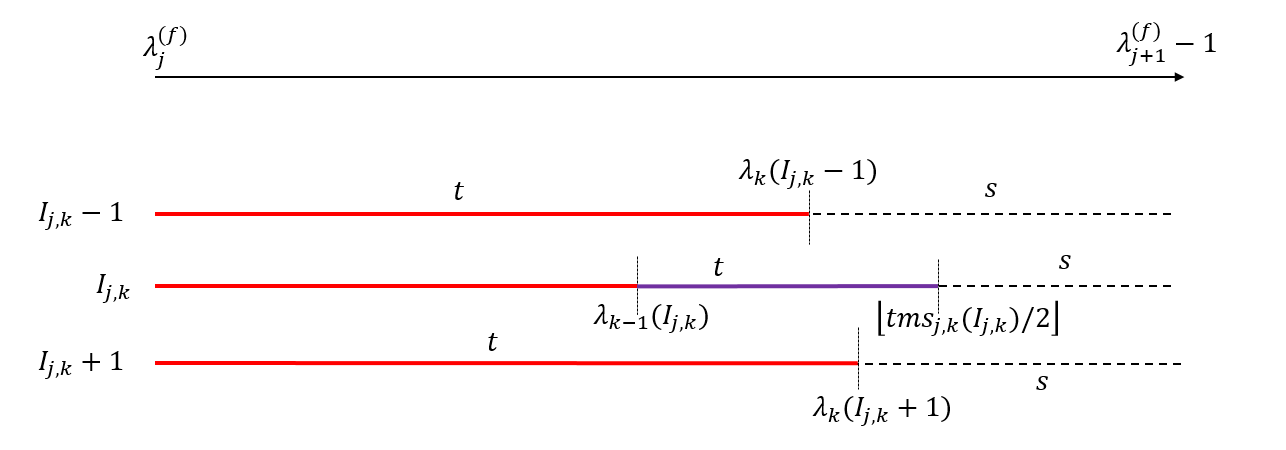}}
    \caption{Illustration of two examples of the case (\ref{eqn: lambda-breakpoint6}). In both examples, for $I_{j,k}$, $smt_{j,k-1}(I_{j,k}) < 0$ in $G_{j,k-1}$ and $smt_{j,k}(I_{j,k}) < 0$ in $G_{j,k}$. The top black arrow denotes the $\lambda$ interval $[\lambda^{(f)}_j, \lambda^{(f)}_{j+1}-1]$. For super-nodes $I_{j,k}-1$, $I_{j,k}$ and $I_{j,k}+1$, the respective solid red line segment denotes the $\lambda$ interval for which the super-node is in $T_{j,k-1}$ in $G_{j,k-1}$. The solid purple line segments denote the new $I_{j,k}$-sink-$\lambda$-intervals introduced in $G_{j,k}$. The remaining dashed line segments denote the source-$\lambda$-intervals in $G_{j,k}$.}\label{figure: smt_negative_2}
    \end{figure}
\end{enumerate}

The above analysis assumes the existence of both $I_{j,k}-1$ and $I_{j,k}+1$. For the corner cases of $I_{j,k} = 1$ ($I_{j,k} - 1$ does not exist) and $I_{j,k} = n_j$ ($I_{j,k} + 1$ does not exist), one can easily check that the results (\ref{eqn: lambda-breakpoint1}) to (\ref{eqn: lambda-breakpoint6}) hold by simply changing $smt_{j,k}(I_{j,k})/2$ to $smt_{j,k}(I_{j,k})$ (thus changing $tms_{j,k}(I_{j,k})/2$ to $tms_{j,k}(I_{j,k})$), with the introduced $\lambda$-breakpoints changed accordingly. Hence the lemma holds for $k$ and we complete the proof.
\end{proof}


Based on Lemma \ref{lem: m_breakpoints}, in each $G_{j,k} (k \in [q])$, at most one $\lambda$-breakpoint is introduced. Therefore Theorem \ref{thm: induced_graph_breakpoints} holds.

As there are at most $n$ reduced PL-FL problems, PL-FL-$\lambda^{(f)}_j$ for $j = 0, 1, \ldots, n-1$, by Theorem \ref{thm: induced_graph_breakpoints}, we immediately have:
\begin{Cor}\label{cor: total-fusing-num}
The total number of $\lambda$-breakpoints in PL-FL (\ref{prob: PL-FL-lambda}) over all nodes for $\lambda \geq 0$ is at most $qn + n - 1$.
\end{Cor}
\begin{proof}
Since each PL-FL-$\lambda^{(f)}_j$ contains at most $q$ $\lambda$-breakpoints, and there are at most $n$ such problems, so the total number is at most $qn$. The additional $n - 1$ accounts for the fusing $\lambda$ values.
\end{proof}

Besides the bound on the number of $\lambda$-breakpoints, from the proof of Lemma \ref{lem: m_breakpoints}, we obtain interesting structure characterizations on the $\lambda$-breakpoints and the path of solutions. For $\lambda$-breakpoints, we immediately have the following corollary:
\begin{Cor}\label{cor: lambda-breakpoint-structure}
In PL-FL (\ref{prob: PL-FL-lambda}), each $\lambda$-breakpoint is of value equals to either (1) $\lceil smt_{j,k}(I_{j,k}) / 2\rceil$ (or $\lceil smt_{j,k}(I_{j,k})\rceil$ if $I_{j,k} = 1$ or $I_{j,k} = n_j$) for some PL-FL-$\lambda^{(f)}_j$ and $k$ if $smt_{j,k}(I_{j,k}) \geq 0$; or (2) $\lfloor tms_{j,k}(I_{j,k})/2\rfloor$ (or $\lfloor tms_{j,k}(I_{j,k}) \rfloor$ if $I_{j,k} = 1$ or $I_{j,k} = n_j$) for some PL-FL-$\lambda^{(f)}_j$ and $k$ if $tms_{j.k}(I_{j,k}) > 0$.
\end{Cor}

To characterize the structure of the path of solutions, we first define a piecewise constant function as \emph{piecewise-constant-quasi-convex} if the following holds:
\begin{Def}
A piecewise constant function $f(x)$ is piecewise-constant-quasi-convex if the list of constant values attained by the function, as $x$ increases, are (i) monotone decrease, or (ii) monotone increase, or (iii) first monotone decrease and then monotone increase.
\end{Def}
Recall from Lemma \ref{lem: piecewise_constant} that the path of solutions are piecewise constant. We have the following local piecewise-constant-quasi-convexity property on the structure of the path of solutions:
\begin{Cor}
In PL-FL-$\lambda^{(f)}_j$, for each super-node $I$ and $\lambda \in [\lambda^{(f)}_j, \lambda^{(f)}_{j+1}-1]$, its optimal solution $x^*_{I}(\lambda)$ as a function of $\lambda$ is piecewise-constant-quasi-convex. Therefore, in PL-FL, for each node $i \in [n]$ and $\lambda \geq 0$, its optimal solution $x^*_i(\lambda)$ as a function of $\lambda$ is locally piecewise-constant-quasi-convex for each $\lambda$-interval $[\lambda^{(f)}_j, \lambda^{(f)}_{j+1} - 1]$.
\end{Cor}
\begin{proof}
Recall that each parametric graph $G_{j,k} = G^{st}_{j}(a_{i_k, j_k})$, where $a_{i_k, j_k}$ is the $k$th largest piecewise linear breakpoint. According to HL-algorithm, a super-node $I$ joins the sink set from source set in $G_{j,k}$ implies that the variable for super-node $I$ attains its optimal value, $a_{i_k, j_k}$. The later a super-node joins the sink set, its optimal value is larger.

For every super-node $I$, its $smt_{j,k}(I)$ value is nonincreasing in $k$ from $G_{j,0}$ to $G_{j,q}$, starting from positive to negative. Based on the derivation of Lemma \ref{lem: m_breakpoints}, if $I$ already joins the sink set at certain $G_{j,k}$ for some values of $\lambda$ where $smt_{j,k}(I) \geq 0$, its sink-$\lambda$-intervals first start from the middle in $[\lambda^{(f)}_{j}, \lambda^{(f)}_{j+1} - 1]$, and then expand to both sides until hitting the endpoints -- in this process the attained optimal values increase on both sides from the middle (bowl shape); when it reaches the stage where $smt_{j,k}(I) < 0$, the expansion of the sink-$\lambda$-intervals must have hit the left endpoint $\lambda^{(f)}_j$, and may further expand to the right until hitting the right endpoint $(\lambda^{(f)}_{j+1}-1)$ -- in this process the attained optimal values further increase to the right. In the other case, if $I$ only joins the sink set for some $\lambda$ values when $smt_{j,k}(I) < 0$, then the sink-$\lambda$-intervals must start from $\lambda^{(f)}_j$ and expand gradually to the right until hitting $\lambda^{(f)}_{j+1}-1$, where the attained optimal values monotonically increase from the left to the right. In both cases, the optimal solution $x^*_I(\lambda)$ is quasi-convex.
\end{proof}

The piecewise-constant-quasi-convexity is illustrated in Figure \ref{figure: sol_structure}.
    \begin{figure}[!h]
    \centering
    \subfloat[The first sink-$\lambda$-interval generated when $smt_{j,k}(I) \geq 0$ in some $G_{j,k}$.\label{subfigure: sol_str_1}]{\includegraphics[scale=0.8]{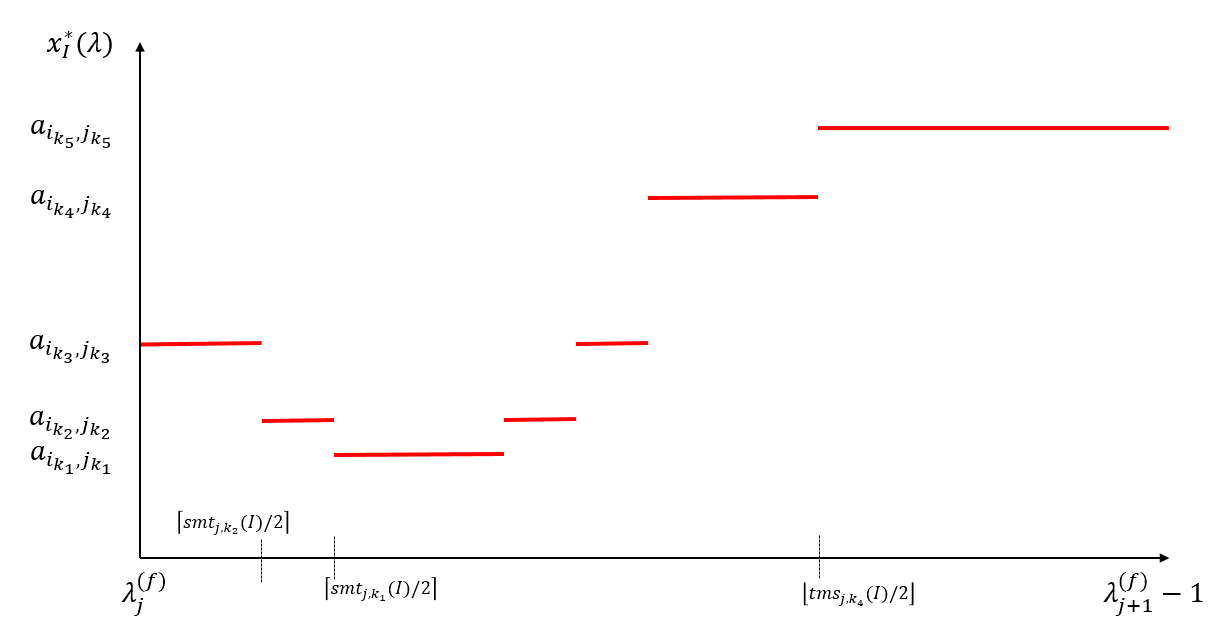}}\\
    \subfloat[The first sink-$\lambda$-interval generated when $smt_{j,k}(I) < 0$ in some $G_{j,k}$.\label{subfigure: sol_str_2}]{\includegraphics[scale=0.8]{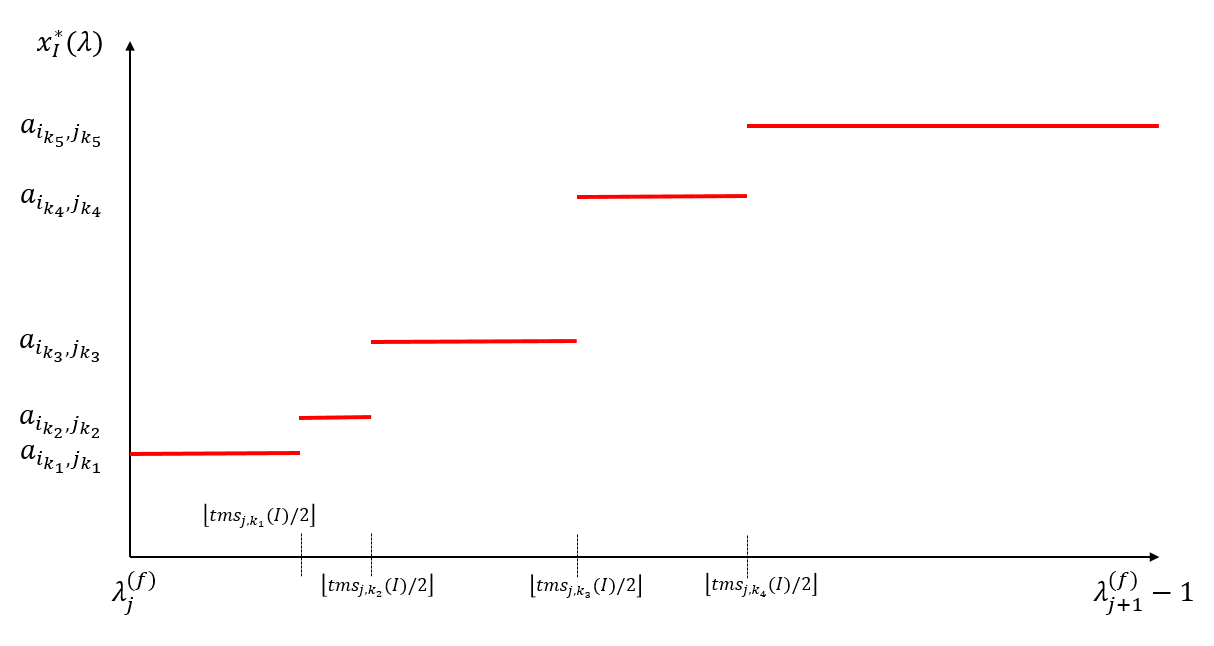}}
    \caption{Illustration of the structure of path of solutions $x^*_I(\lambda)$ for a super-node $I$ for $\lambda \in [\lambda^{(f)}_j, \lambda^{(f)}_{j+1}-1]$. The two cases are shown separately. Each horizontal line corresponds to one source-to-sink transition in one parametric graph. In both figures, $k_1 < k_2 < k_3 < k_4 < k_5$, so $a_{i_{k_1}, j_{k_1}} < a_{i_{k_2}, j_{k_2}} < a_{i_{k_3}, j_{k_3}} < a_{i_{k_4}, j_{k_4}} < a_{i_{k_5}, j_{k_5}}$. Newly introduced $\lambda$-breakpoints are shown on the horizontal axes. Note that in figure (\ref{subfigure: sol_str_1}), $\lceil smt_{j,k_3}(I)/2\rceil \leq \lambda^{(f)}_j$ and $\lfloor tms_{j,k_5}(I)/2\rfloor \geq \lambda^{(f)}_{j+1}-1$, so the two values are not introduced as $\lambda$-breakpoints. Similarly, in figure (\ref{subfigure: sol_str_2}), $\lfloor tms_{j,k_5}(I)/2 \rfloor \geq \lambda^{(f)}_{j+1}-1$, so is not introduced as $\lambda$-breakpoint either.}\label{figure: sol_structure}
    \end{figure}

\subsection{FL (\ref{prob: fused-lasso-lambda})}
Recall from the introduction section that an FL (\ref{prob: fused-lasso-lambda}) of solution accurary $\epsilon$ is equivalent to a PL-FL (\ref{prob: PL-FL-lambda}) of total number of piecewise linear breakpoints $q = O(\frac{nU}{\epsilon})$, where $U = \max_i\{u_i - \ell_i\}$. As a result, applying Corollary \ref{cor: total-fusing-num}, we immediately have:
\begin{Thm}
For FL (\ref{prob: fused-lasso-lambda}) of general convex loss functions, the total number of $\lambda$-breakpoints over all nodes for $\lambda \geq 0$ is at most $O(\frac{n^2U}{\epsilon} + n) = O(\frac{n^2U}{\epsilon})$, where $U = \max_i\{u_i - \ell_i\}$.
\end{Thm}

Next we focus on the algorithm to solve the path of solutions to PL-FL (\ref{prob: PL-FL-lambda}).


\section{Algorithm to solve the path of solutions to PL-FL (\ref{prob: PL-FL-lambda})}\label{sect: path_alg}
To solve the path of solutions for PL-FL (\ref{prob: PL-FL-lambda}), we first find all fusing $\lambda$ values via a binary search method, by which we generate the reduced PL-FL problems. Then for each PL-FL-$\lambda^{(f)}_j$ problems, for $\lambda\in [\lambda^{(f)}_j, \lambda^{(f)}_{j+1} - 1]$, we find all $\lambda$-breakpoints of all the super-nodes. In the process, for each super-node, we obtain all its maximal $\lambda$-constant-intervals and their corresponding optimal values.

\subsection{Base Data Structure}\label{sect: red-black-tree}
The base data structure employed in our algorithm to store the (intermediate) results is red-black tree \cite{CLRS09}. A red-black tree is a binary search tree. Each node of the tree contains the following five fields \cite{CLRS09}:\\
{\underline{\emph{color}:}} The ``color" of a node. Its value is either RED or BLACK.\\
{\underline{\emph{key}:}} The ``key" value of a node. It is a scalar.\\
{\underline{\emph{left, right}:}} The pointers to the left and the right child of a node. If the corresponding child does not exist, the corresponding pointer has value NIL.\\
{\underline{\emph{p}:}} The pointer to the parent of a node. If the node is the root node, the pointer value is NIL.

As it is a binary search tree, the keys of the nodes are comparable. Furthermore, it has the following two properties \cite{CLRS09}:
\begin{enumerate}
\item Binary-search-tree property: Let $x$ be a node in a binary search tree. If $y$ is a node in the left subtree of $x$, then $key[y] \leq key[x]$. If $y$ is a node in the right subtree of $x$, then $key[y] \geq key[x]$.
\item Tree height property: A red-black tree with $n$ nodes has height at most $2\log(n+1)$.
\end{enumerate}

Cormen et al. in \cite{CLRS09} define and analyze the following three operations on a red-black tree:
\begin{enumerate}
\item {\sf TREE-SEARCH}($T$, $k$): Search for a node in red-black tree $T$ with a given \emph{key} value $k$. It returns a pointer to a node with \emph{key} $k$ if one exists; otherwise it return NIL. 
\item {\sf RB-INSERT}($T$,$z$): Insert a node $z$ into red-black tree $T$. 
\item {\sf RB-DELETE}($T$, $z$): Delete a node $z$ from red-black tree $T$. 
\end{enumerate}
Cormen et al. in \cite{CLRS09} prove that each of the above operation has complexity $O(\log n)$ for a tree of at most $n$ nodes. 

We will extend the above base form of red-black tree while maintaining all the above complexity results.

\subsection{Compute fusing $\lambda$ values}\label{sect: binary-search}
We maintain a red-black tree $T_f$ in the search for all fusing $\lambda$ values with the following extension. The \emph{key} fields in $T_f$ are the integer $\lambda$ value. In addition, each node in $T_f$ of \emph{key} $\lambda$ contains an $n$-bit array $group_{\lambda}$, which contains the fused group information of the optimal solution for the $\lambda$ value. The $i$th bit in $group_{\lambda}$, $group_{\lambda}(i)$, corresponds to variable $x_i$ in PL-FL (\ref{prob: PL-FL-lambda}). The bits in $group_{\lambda}$ is defined as follows:
\begin{equation*}
group_{\lambda}(1) = 1,\ group_{\lambda}(i) = \begin{cases}
									group_{\lambda}(i-1),\ &\text{if}\ x^*_i(\lambda) = x^*_{i-1}(\lambda),\\
									1 - group_{\lambda}(i-1),\ &\text{if}\ x^*_i(\lambda) \neq x^*_{i-1}(\lambda).
									\end{cases}\ (\forall\ i \geq 2)
\end{equation*}
We say that $group_\lambda = group_{\lambda'}$ if $group_{\lambda}(i) = group_{\lambda'}(i)$ for all $i$, otherwise $group_\lambda \neq group_{\lambda'}$. $group_{\lambda} = group_{\lambda'}$ implies that the optimal solutions for $\lambda$ and $\lambda'$ have the same fused groups.

Given the optimal solution for a $\lambda$, generating the values for $group_{\lambda}$ takes an additional $O(n)$ time. Given two group values $group_{\lambda}$ and $group_{\lambda'}$, it also takes $O(n)$ time to check whether the two groups are equal. Note that as the group values are bit strings, checking whether $group_{\lambda} = group_{\lambda'}$ is equivalent to checking that whether bitwise XOR of the two bit strings is equal to 0. Bit operations can be done very efficiently in computers.

The algorithm maintains that each $\lambda$ value of a node in $T_f$ is a candidate fusing $\lambda$ value. As there are $O(n)$ fusing $\lambda$ values, the number of nodes in $T_f$ is $O(n)$. In the algorithm, we apply the following operations to $T_f$:
\begin{enumerate}
\item $z := {\sf new\_node}(\lambda, group_{\lambda})$: Create a new node $z$ with \emph{key} $\lambda$ and group array $group_\lambda$. This is done in $O(1)$ time.
\item {\sf TREE-SEARCH}$(T_f, \lambda)$: Search for the node in red-black tree $T_f$ with a given \emph{key} value $\lambda$. It returns a pointer to the node with \emph{key} $\lambda$ if one exists; otherwise it returns NIL. This operation can be done in time $O(\log n)$ for $T_f$ of at most $O(n)$ nodes.
\item {\sf RB-INSERT}$(T_f, z)$: Insert a node $z$ into red-black tree $T$. This can also be done in complexity $O(\log n)$ for $T_f$ of at most $O(n)$ nodes.
\end{enumerate}

We first introduce the binary search algorithm to find all fusing $\lambda$ values in an interval $[\lambda_\ell, \lambda_u]$ (assuming $group_{\lambda_\ell} \neq group_{\lambda_u}$).
we compute the optimal solution of PL-FL (\ref{prob: PL-FL-lambda}) for $\lambda_{m} = \lfloor(\lambda_{\ell} + \lambda_u)/2\rfloor$ and get $group_{\lambda_{m}}$. Then we compare $group_{\lambda_{m}}$ with $group_{\lambda_{\ell}}$ and $group_{\lambda_u}$. There are three possibilities:
\begin{enumerate}
\item $group_{\lambda_{m}} \neq group_{\lambda_{\ell}}$ and $group_{\lambda_{m}} \neq group_{\lambda_{u}}$:\\
    A new node with \emph{key} value $\lambda_{m}$ and group array $group_{\lambda_{m}}$ is inserted into $T_f$. The search for fusing $\lambda$ values continues in the intervals $[\lambda_{\ell}, \lambda_{m}]$ and $[\lambda_{m} , \lambda_u]$.
\item $group_{\lambda_{m}} = group_{\lambda_{\ell}}$ but $group_{\lambda_{m}} \neq group_{\lambda_{u}}$:\\
    No update to $T_f$ and ignore the interval $[\lambda_\ell, \lambda_{m}]$, because there will be no additional fusing variables for $\lambda$ in the interval. The search for fusing $\lambda$ values continues in the interval $[\lambda_{m}, \lambda_u]$.
\item $group_{\lambda_{m}} \neq group_{\lambda_{\ell}}$ but $group_{\lambda_{m}} = group_{\lambda_{u}}$:\\
    Update the node with \emph{key} $\lambda_u$ in $T_f$ to the new $key$ value $\lambda_{m}$ while no change to the group array in the node as $group_{\lambda_{m}} = group_{\lambda_u}$. Ignore the interval $[\lambda_{m}, \lambda_u]$ because there will be no additional fusing variables for $\lambda$ in the interval. The search for fusing $\lambda$ values continues in the interval $[\lambda_\ell, \lambda_{m}]$.
\end{enumerate}

The pseudo-code is as follows:

\vspace{1em}

\noindent$T_f$ := {\sf search\_fusing\_values}$(\lambda_\ell, group_{\lambda_\ell}, \lambda_u, group_{\lambda_u}, T_f)$\\
\noindent\textbf{begin}\\
1\ \taba \textbf{if} $\lambda_u - \lambda_\ell \leq 1$ \textbf{then} exit; \textbf{end if}\\
2\ \taba \textbf{if} $group_{\lambda_{\ell}} = group_{\lambda_u}$ exit; \textbf{end if}\\
3\ \taba $\lambda_{m} := \lfloor\frac{\lambda_u + \lambda_\ell}{2}\rfloor$;\\
4\ \taba Solve PL-FL (\ref{prob: PL-FL-lambda}) for $\lambda = \lambda_{m}$, compute $group_{\lambda_{m}}$;\\
5\ \taba \textbf{if} $group_{\lambda_{m}} \neq group_{\lambda_\ell}$ and $group_{\lambda_{m}} \neq group_{\lambda_u}$ \textbf{then} \\
6\ \tabb $z := {\sf new\_node}(\lambda_{m}, group_{\lambda_m})$;\\
7\ \tabb {\sf RB-INSERT}$(T_f, z)$;\\
8\ \tabb $T_f$ := {\sf search\_fusing\_values}$(\lambda_{\ell}, group_{\lambda_\ell}, \lambda_{m}, group_{\lambda_{m}}, T_f)$; \\
9\ \tabb $T_f$ := {\sf search\_fusing\_values}$(\lambda_{m}, group_{\lambda_{m}}, \lambda_u, group_{\lambda_u}, T_f)$;\\
10\taba \textbf{else if} $group_{\lambda_{m}} = group_{\lambda_\ell}$ and $group_{\lambda_{m}} \neq group_{\lambda_u}$ \textbf{then} \\
11\tabb $T_f$ := {\sf search\_fusing\_values}$(\lambda_{m}, group_{\lambda_{m}}, \lambda_u, group_{\lambda_u}, T_f)$;\\
12\taba \textbf{else if} $group_{\lambda_{m}} \neq group_{\lambda_\ell}$ and $group_{\lambda_{m}} = group_{\lambda_u}$ \textbf{then} \\
13\tabb $z := ${\sf TREE-SEARCH}$(T_f, \lambda_u)$;\\
14\tabb $z.key := \lambda_{m}$;\\
15\tabb $T_f$ := {\sf search\_fusing\_values}$(\lambda_\ell, group_{\lambda_\ell}, \lambda_{m}, group_{\lambda_{m}}, T_f)$;\\
16\taba \textbf{end if}\\
\noindent\textbf{end}

\vspace{1em}

To compute all fusing $\lambda$ values in $[0, +\infty)$, we first replace the right endpoint from $+\infty$ to some $\lambda_{\max}$ such that PL-FL for $\lambda_{\max}$ has optimal solution where all variables are fused together. One feasible value for $\lambda_{\max}$ is
\begin{equation}\label{eqn: lambda-max}
\lambda_{\max} = \biggl\lceil \frac{\sum_{i=1}^n f^{pl}_i(0) - \sum_{i=1}^n \min_{x_i}f^{pl}_i(x_i)}{a_{\min, +}}\biggr\rceil + 1,
\end{equation}
where $\sum_{i=1}^n f^{pl}_i(0)$ is a feasible value for PL-FL (\ref{prob: PL-FL-lambda}), $\sum_{i=1}^n \min_{x_i}f^{pl}_i(x_i)$ is a lower bound of the optimal value for PL-FL, and
\begin{equation*}
a_{\min, +} = \min_{\substack{k = 2, \ldots, q,\\a_{i_k, j_k} > a_{i_{k-1}, j_{k-1}}}}\{a_{i_{k}, j_{k}} - a_{i_{k-1}, j_{k-1}}\},
\end{equation*}
which is the minimum positive distance among all the piecewise linear breakpoints of the loss functions. It is easy to verified that this $\lambda_{\max}$ value forces the optimal values of all variables in PL-FL to be the same. Recall that each loss function $f^{pl}_i(x_i)$ is represented by its piecewise linear breakpoints in ascending order and the slopes of the linear pieces in-between. Hence the complexity to compute the above $\lambda_{\max}$ value is $O(q)$.

The pseudo-code to compute all fusing $\lambda$ values is the following {\sf find\_all\_fusing\_values}. It returns a sorted list of fusing $\lambda$ values with associated $group_{\lambda}$ arrays.

\vspace{1em}

\noindent$\big(\lambda^{(f)}_j, group_{\lambda^{(f)}_j}\big)_{j = 0,1,\ldots, p} := {\sf find\_all\_fusing\_values}()$\\
\noindent\textbf{begin}\\
\taba Initialize an empty red-black tree $T_f$;\\
\taba Compute $\lambda_{\max}$ according to Equation (\ref{eqn: lambda-max});\\
\taba Solve PL-FL (\ref{prob: PL-FL-lambda}) for $\lambda = 0$, compute $group_0$;\\
\taba $z := {\sf new\_node}(0, group_0)$; {\sf RB-INSERT}$(T_f, z)$; \\
\taba Solve PL-FL (\ref{prob: PL-FL-lambda}) for $\lambda = \lambda_{\max}$, compute $group_{\lambda_{\max}}$; \\
\taba \textbf{if} $group_{\lambda_{\max}} \neq group_0$ \textbf{then} \\
    \tabb $z := {\sf new\_node}(\lambda_{\max}, group_{\lambda_{\max}})$; {\sf RB-INSERT}$(T_f, z)$; \\
    \tabb $T_f$ := {\sf search\_fusing\_values}($0, group_0, \lambda_{\max}, group_{\lambda_{\max}}, T_f$);\\
\taba \textbf{end if}\\
\taba In-order traversal on $T_f$ to \textbf{return} $\biggl(\lambda^{(f)}_j, group_{\lambda^{(f)}_j}\biggr)_{j = 0,1,\ldots, p}$;\\
\noindent\textbf{end}

\vspace{1em}

The correctness of the algorithms is justified by Theorem \ref{thm: fusing-points}. We analyze the complexity of the two pseudo-codes. First {\sf search\_fusing\_values}. As the number of fusing $\lambda$ values is $O(n)$, the number of times of the case $group_{\lambda_{m}} \neq group_{\lambda_{\ell}}$ and $group_{\lambda_{m}} \neq group_{\lambda_u}$ (line 5) happening is $O(n)$. Between two consecutive fusing $\lambda$ values, the computed $group_{\lambda_m}$ must fall in either of the latter two cases in the if/else-statement (line 10 to line 16), for which at every iteration the search interval is cut by half. As a result, the number of trial $\lambda$ values in {\sf search\_fusing\_values} is $O(n\log(\lambda_u - \lambda_\ell))$. For each trial $\lambda$ value, the algorithm first solves PL-FL (\ref{prob: PL-FL-lambda}) and compute $group_\lambda$ at line 4. Let $T_0$ be the time complexity to solve PL-FL for a fixed $\lambda$. Thus the complexity of line 4 is $O(T_0 + n)$.
To proceed with the if/else-statement, the code compares $group_{\lambda_m}$ with $group_{\lambda_\ell}$ and $group_{\lambda_u}$, which incurs an additional $O(n)$ time. Each block of the if/else-statement is at most $O(\log n)$. As a result, the total computation complexity for each trial $\lambda$ value is $O(T_0 + n + \log n)$. Therefore, the total complexity of {\sf search\_fusing\_values} is $O(n\log(\lambda_u - \lambda_\ell)(T_0 + n + \log n))$.

The complexity of {\sf find\_all\_fusing\_values} is dominated by {\sf search\_fusing\_values} and the computation of $\lambda_{\max}$, thus its complexity is $O((n\log \lambda_{\max})(T_0 + n + \log n) + q)$. 

For PL-FL (\ref{prob: PL-FL-lambda}) of fixed $\lambda$, the fastest algorithm is HL-algorithm of complexity $T_0 = O(q\log n)$ by \cite{HL17}. Therefore the complexity of {\sf find\_all\_fusing\_values} is $O(nq\log n\log\lambda_{\max})$ ($q = \Omega(n)$).


\subsection{Solve PL-FL-$\lambda^{(f)}_j$ for $\lambda \in [\lambda^{(f)}_j, \lambda^{(f)}_{j+1}-1]$}
With the fusing $\lambda$ values and the fusing group arrays obtained, we can generate all the reduced PL-FL problems. Next we solve the path of solutions of PL-FL-$\lambda^{(f)}_j$ for $\lambda \in [\lambda^{(f)}_j, \lambda^{(f)}_{j+1}-1]$.

\subsubsection{Data structures}
In PL-FL-$\lambda^{(f)}_j$, we store the path of solutions of each super-node $I$ by a red-black tree $T_{j,I}$ with the following extension from the basic red-black tree in Section \ref{sect: red-black-tree}. The \emph{key} field of each node in $T_{j,I}$ is extended from a scalar to a 2-tuple, $\lambda_\ell < \lambda_r$, which represents a maximal $\lambda$-constant-interval $[\lambda_\ell, \lambda_r]$, and the node has a \emph{value} field that is the associated constant optimal value of $x_{I}$ for $\lambda \in [\lambda_\ell, \lambda_r]$. The comparison of the \emph{key tuples} of nodes in $T_{j,I}$ follows from the comparison of their respective $\lambda$-intervals defined in Section \ref{sect: path_structure} following Lemma \ref{lem: piecewise_constant}. According to Theorem \ref{thm: induced_graph_breakpoints}, the number of nodes in each $T_{j,I}$ is $O(q)$.

The extension of red-black trees from scalar keys to tuple keys is also employed in HL-algorithm in \cite{HL17}, where they use the following four operations with complexities shown:
\begin{enumerate}
\item $z := {\sf new\_node}(\lambda_\ell, \lambda_r, a)$: Create a new node $z$ with \emph{key tuple} $key[z].first = \lambda_\ell, key[z].right = \lambda_r$ and $value[z] = a$. This is done in $O(1)$ time.
\item $[\lambda_\ell, \lambda_r] := {\sf get\_\lambda\_interval}(T_{j,I}, \lambda)$: Find the maximal $\lambda$-constant-interval $[\lambda_\ell, \lambda_r]$ in $T_{j,I}$ that contains the given $\lambda$ value. This is done in $O(\log q)$ time for $T_{j,I}$ of at most $O(q)$ nodes.
\item {\sf TREE-SEARCH}$(T_{j,I}, [\lambda_\ell, \lambda_r])$: Search for the node in red-black tree $T_{j,I}$ with given \emph{key tuple} generated from an $\lambda$-interval $[\lambda_\ell, \lambda_r]$. This is done in $O(\log q)$ time for $T_{j,I}$ of at most $O(q)$ nodes.
\item {\sf RB-INSERT}$(T_{j,I}, z)$: Insert a node $z$ into $T_{j,I}$. This is done in $O(\log q)$ time for $T_{j,I}$ of at most $O(q)$ nodes.
\end{enumerate}
Our algorithm presented here will apply the above four operations to $T_{j,I}$. Initially, red-black trees $T_{j,I}$ for all super-nodes $I$ are empty.

Recall that PL-FL-$\lambda^{(f)}_j$ is generated from PL-FL (\ref{prob: PL-FL-lambda}) by fusing nodes of same optimal value into a super-node. From the $(\lambda^{(f)}_j, group_{\lambda^{(f)}_j})$, we create a table for mapping between a node in PL-FL and its corresponding super-node in PL-FL-$\lambda^{(f)}_j$, and vice versa. To create the table, one only needs to traverse $group_{\lambda^{(f)}_j}$ once, in time $O(n)$. Let the table be $TB_j$, where $I = TB_j(i) (i = 1,\ldots, n)$ returns the super-node $I$ in PL-FL-$\lambda^{(f)}_j$ that corresponds to node $i$ in PL-FL.


For each super-node $I$ in PL-FL-$\lambda^{(f)}_j$, its corresponding piecewise linear loss function is generated by summing up the piecewise linear functions of its containing nodes in PL-FL. For our algorithm purpose, we do not need to merge and sort the sub-lists of the piecewise linear breakpoints of the fused piecewise linear loss functions because our algorithm always traverse the full list of all the piecewise linear breakpoints in ascending order. From the analysis in Section \ref{sect: path_structure}, we introduce an array $(smt_j(I))_{I=1,\ldots, n_j}$ for the quantities $\{smt_{j,k}(I)\}_{k=1,\ldots,q; I = 1,\ldots, n_j}$. The array $(smt_j(I))_{I=1,\ldots,n_j}$ is updated throughout the algorithm such that for every super-node $I$, in $G_{j,k}$, $smt_j(I) = smt_{j,k}(I)$. The array $(smt_j(I))_{I=1,\ldots,n_j}$ is updated as follows: 
\begin{enumerate}
\item Initially, in $G_{j,0}$:
\begin{equation}\label{eqn: smt_init}
smt_j(I) = c_{s,I} - c_{I,t} = c_{s,I} = \sum_{i: TB_j(i) = I} -w_{i,0} \geq 0.
\end{equation}
The array $(smt_j(I))$ for all super-nodes $I$ can be initiated by traversing the nodes from $1$ to $n$ in PL-FL once, which has $O(n)$ complexity\footnote{In practice, one can speed up the initialization of the array for every reduced PL-FL problem by introducing a global partial-sum array $(sa(i))_{i=0,\ldots,n}$ for PL-FL (\ref{prob: PL-FL-lambda}) as follows: $sa(0) = 0$, $sa(i) = sa(i - 1) - w_{i,0}, i = 1, \ldots, n$. Then for each super-node $I = [i_\ell, i_r]$ in PL-FL-$\lambda^{(f)}_j$, $smt_j(I) = sa(i_r) - sa(i_\ell - 1)$.}.
\item $G_{j,k-1}$ to $G_{j,k}$: Only the source and sink adjacent arc capacities of super-node $I_{j,k}$ change. The right sub-gradient of convex piecewise linear function $f^{pl}_{I_{j,k}}$ changes by the amount $\Delta_{k} = w_{i_k,j_k} - w_{i_k, j_{k-1}} > 0$. 
One can verify that we have:
\begin{equation*}
smt_j(I_{j,k}) := smt_j(I_{j,k}) - \Delta_{k}.
\end{equation*}
The update is done in $O(1)$ time.
\end{enumerate}
For convenience of presentation, we also an array $(tms_j(I))_{I = 1,\ldots,n_j}$ such that $tms_j(I) = -smt_j(I)$.


Our algorithm follows the analysis in Proposition \ref{prop: lambda_ranges} and Lemma \ref{lem: m_breakpoints}. Following Lemma \ref{lem: m_breakpoints}, we define a tuple array for each super-node $I$, $(sink\_intv_j(I))_{I=1,\ldots, n_j}$, such that after processing $G_{j,k}$, the $\lambda$-interval $[sink\_intv_j(I).first, sink\_intv_j(I).second]$ is the maximal $I$-sink-$\lambda$-interval. According to Lemma \ref{lem: m_breakpoints}, if $smt_{j,k}(I) \geq 0$,
\begin{equation*}
sink\_intv_j(I).first = \lambda_{k,\ell}(I),\ sink\_intv_j(I).second = \lambda_{k,r}(I);
\end{equation*}
if $smt_{j,k}(I) < 0$,
\begin{equation*}
sink\_intv_j(I).first = \lambda^{(f)}_j,\ sink\_intv_j(I).second = \lambda_k(I).
\end{equation*}
Initially, in $G_{j,0}$, $sink\_intv_j(I).first = \lambda^{(f)}_j, sink\_intv_j(I).second = \lambda^{(f)}_{j}-1$, \ie, $sink\_intv_j(I) = \emptyset$, for all super-nodes $I$.



\subsubsection{Algorithm}

The algorithm directly follows Lemma \ref{lem: m_breakpoints}, with data $(smt_j(I), sink\_intv_j(I), T_{j,I})_{I = 1,\ldots,n_j}$ updated when the algorithm processes from $G_{j,0}$ to $G_{j,q}$. We present the pseudo-code to solve the path of solutions of PL-FL-$\lambda^{(f)}_j$ for $\lambda \in [\lambda^{(f)}_j, \lambda^{(f)}_{j+1}-1]$ as follows:

\vspace{1em}

\noindent$\{(TB_j(i))_{i=1,\ldots,n}, (T_{j,I})_{I = 1,\ldots,n_j}\} := ${\sf solve\_reduced\_PL-FL}$(\lambda^{(f)}_j, \lambda^{(f)}_{j+1}, group_{\lambda^{(f)}_j})$\\
\noindent\textbf{begin}\\
1\ \taba Compute $n_j$ and $(TB_j(i))_{i=1,\ldots,n}$ from $group_{\lambda^{(f)}_j}$;\\
2\ \taba Initialize the array $(smt_j(I))_{I = 1,\ldots,n_j}$ according to (\ref{eqn: smt_init});\\
3\ \taba Initialize red-black trees $T_{j,I}$ to be empty and $sink\_intv_j(I) = \emptyset$ for $I = 1,\ldots,n_j$;\\
4\ \taba \textbf{for} $k := 1,\ldots,q$:\\
5\    \tabb $I_{j,k} := TB_j(i_k)$; \\
6\    \tabb \{Update graph\} $smt_j(I_{j,k}) := smt_j(I_{j,k}) - (w_{i_k,j_k} - w_{i_k,j_{k-1}})$;\\
7\    \tabb $(T_{j,I_{j,k}}, sink\_intv_j(I_{j,k})) := {\sf compute\_\lambda\_breakpoint}(I_{j,k}, n_j, \lambda^{(f)}_j, \lambda^{(f)}_{j+1}, smt_j(\cdot), a_{i_k, j_k},\\
\tabi \tabf T_{j,I_{j,k}}, sink\_intv_j(\cdot))$; \\
8\ \taba \textbf{end for}\\
9\ \taba \textbf{return} $(T_{j,I})_{I=1,\ldots,n_j}$;\\
\noindent\textbf{end}

\vspace{1em}

The initialization of the data structures is done from line 1 to line 3. At line 4, the \textbf{for} loop computes, in the $k$th iteration, the minimum cut in $G_{j,k}$ from the minimum cut in $G_{j,k-1}$. The super-node $I_{j,k}$ is obtained from the table $TB_j(i_k)$ at line 5. Then in line 6, the value of $smt_j(I_{j,k})$ is updated from $G_{j,k-1}$ to $G_{j,k}$. Line 7 follows the analysis in Lemma \ref{lem: m_breakpoints}, where (at most one) new $\lambda$-breakpoint and (at most two) new maximal $\lambda$-constant-intervals, together with the corresponding optimal value $a_{i_k, j_k}$, could be introduced for $I_{j,k}$, and thus the data $T_{j,I_{j,k}}$ and $sink\_intv_j(I_{j,k})$ are potentially updated. The detailed implementation of {\sf compute\_$\lambda$\_breakpoint} is in Appendix \ref{appx: compute}.

In Appendix \ref{appx: compute}, we show that each call to subroutine {\sf compute\_$\lambda$\_breakpoint} takes $O(\log q)$ time.
As a result, the total complexity of the \textbf{for} loop from line 4 to line 8 is $O(q\log q)$. The initialization steps from line 1 to line 3 has complexity $O(n)$. Therefore, the total complexity of {\sf solve\_reduced\_PL-FL} is $O(q\log q + n) = O(q\log q)$ as $q = \Omega(n)$.

\subsection{Complete algorithm}
With the above subroutines discussed, we are ready to present the complete algorithm to solve the path of solutions of PL-FL (\ref{prob: PL-FL-lambda}) for $\lambda \geq 0$. The pseudo-code of the complete algorithm is as follows:

\vspace{1em}

\noindent{\sf solve\_PL-FL\_solution\_path}\\
\noindent\textbf{input}: $\{\{a_{i,1}, \ldots, a_{i,q_i}\}, \{w_{i,0},\ldots,w_{i,q_i}\}\}_{i=1,\ldots,n}$.\\
\noindent\textbf{output}: $\big\{\lambda^{(f)}_j, group_{\lambda^{(f)}_j}, (TB_j(i))_{i=1,\ldots,n}, \{T_{j,I}\}_{I=1,\ldots,n_j}\big\}_{j=0,1,\ldots,p}$.\\
\noindent\textbf{begin}\\
\taba $(\lambda^{(f)}_j, group_{\lambda^{(f)}_j})_{j = 0,1,\ldots, p} := {\sf find\_all\_fusing\_values}()$;\\
\taba $\lambda^{(f)}_{p+1} = \lambda^{(f)}_p+1$; \\
\taba Sort the breakpoints as $a_{i_1,j_1} < a_{i_2,j_2} < \ldots < a_{i_q,j_q}$;\\
\taba \textbf{for} $j := 0,\ldots,p$:\\
    \tabb $\{(TB_j(i))_{i=1,\ldots,n}, (T_{j,I})_{I=1,\ldots,n_j}\} := ${\sf solve\_reduced\_PL-FL}$(\lambda^{(f)}_j, \lambda^{(f)}_{j+1}, group_{\lambda^{(f)}_j})$;\\
\taba \textbf{end for}\\
\taba \textbf{return} $\big\{\lambda^{(f)}_j, group_{\lambda^{(f)}_j}, (TB_j(i))_{i=1,\ldots,n}, \{T_{j,I}\}_{I=1,\ldots,n_j}\big\}_{j=0,1,\ldots,p}$;\\
\noindent\textbf{end}

\vspace{1em}

The complexity of {\sf find\_all\_fusing\_values} is $O(nq\log n\log\lambda_{\max})$, the complexity of all the calls to {\sf solve\_reduced\_PL-FL} is $O(pq\log q) = O(nq\log q)$ as $p = O(n)$, and the complexity of sorting the breakpoints from $n$ sorted sub-lists is $O(q\log n)$ \cite{HL17}, therefore the total complexity of {\sf solve\_PL-FL\_solution\_path} is $O(nq(\log n\log\lambda_{\max} + \log q)) = \tilde{O}(nq)$.


\subsection{Discussions}\label{sect: discussion}
The path of solutions is stored in the tuple $\{\lambda^{(f)}_j, group_{\lambda^{(f)}_j}, (TB_j(i))_{i=1,\ldots,n}, \{T_{j,I}\}_{I=1,\ldots,n_j}\}_{j=0,1,\ldots,p}$. The space complexity is $O(n(1 + n + n + nq)) = O(n^2q)$.

Given the above encoded path of solutions, we can solve the optimal solution of PL-FL (\ref{prob: PL-FL-lambda}) for any given $\lambda$ value efficiently. We first do a binary search on all fusing $\lambda$ values to find the interval $[\lambda^{(f)}_j, \lambda^{(f)}_{j+1}-1]$ that contains the $\lambda$ value. This has complexity $O(\log n)$. Then for fused group in $group_{\lambda^{(f)}_j}$, we arbitrarily pick on node $i$ $(i\in [n])$ and compute the super-node $I := TB_j(i)$. Then we find in $T_{j,I}$ the node whose maximal $\lambda$-constant-interval contains the $\lambda$ value. This is done in $O(\log q)$ time by calling the subroutine {\sf get\_$\lambda$\_interval} on $T_{j,I}$. The \emph{value} field of the found node in $T_{j,I}$ is the optimal value of $x_i$ for all nodes $i$ in the fused group. The complexity of this procedure is $O(\log n + n_j\log q) = O(n_j \log q)$. It is much faster than solving it from scratch using HL-algorithm in \cite{HL17} of complexity $O(q\log n)$. As a result, using the generated path of solutions, solving PL-FL (\ref{prob: PL-FL-lambda}) of $K$ different $\lambda$ values has worst total complexity $O(nq(\log n\log\lambda_{\max} + \log q) + Kn\log q)$, while solving PL-FL (\ref{prob: PL-FL-lambda}) from scratch for each $\lambda$ using HL-algorithm has complexity $O(Kq\log n)$. Therefore if $K = \Omega(n)$, using the path of solutions gives a faster algorithm.

The encoding of the path of solutions using red-black trees $T_{j,I}$ facilitates the search of optimal solution for a given $\lambda$ value. One can add an additional data structure for the path of solutions that facilitates the search of $\lambda$ values for a given optimal solution. For each PL-FL-$\lambda^{(f)}_j$ problem, we introduce $q$ lists $\{L_{j,k}\}_{k = 1,\ldots,q}$ such that $L_{j,k}$ stores the sorted maximal $\lambda$-constant-intervals in $[\lambda^{(f)}_j, \lambda^{(f)}_{j+1}-1]$ whose optimal solution of super-node $I_{j,k}$ is $a_{i_k, j_k}$. The $L_{j,k}$ array is created in subroutine {\sf update\_$\lambda$\_breakpoint} (see Appendix \ref{appx: compute}): At $G_{j,k}$, if there are (at most two) new maximal $\lambda$-constant-intervals of optimal value $a_{i_k, j_k}$ generated for super-node $I_{j,k}$, these maximal $\lambda$-constant-intervals form $L_{j,k}$. It only incurs an additional $O(1)$ complexity to {\sf update\_$\lambda$\_breakpoint} subroutine.

With the $L_{j,k}$ arrays, we can solve the following inverse optimization problem: Given a stretch of nodes $i_\ell,i_\ell+1,\ldots,i_r-1,i_r$, identify all $\lambda$ values such that $x^*_{i_\ell} = x^*_{i_\ell+1} = \ldots = x^*_{i_r-1} = x^*_{i_r} = a_{i_k,j_k}$ in PL-FL (\ref{prob: PL-FL-lambda}), or output {\sf NULL} if no such $\lambda$ exists. To solve this problem, we first identify the smallest fusing $\lambda$ value, say $\lambda^{(f)}_{j_0}$, such that $x_{i_\ell}$ to $x_{i_r}$ have the same optimal value and $[i_\ell, i_\ell+1,\ldots,i_r-1,i_r] \subseteq I_{j_0,k} (:= TB_{j_0}(i_k))$. Then we check each $L_{j, k}$ for $\lambda^{(f)}_j \geq \lambda^{(f)}_{j_0}$. If $L_{j,k} \neq \emptyset$, then the $\lambda$ values in the maximal $\lambda$-constant-intervals in $L_{j,k}$ are part of the solution. From the solution set, we can also answer questions like the minimum and maximum values of $\lambda$ that achieve the optimal solution. Identifying the $\lambda^{(f)}_{j_0}$ fusing value can be done via binary search on $\{\lambda^{(f)}_j, group_{\lambda^{(f)}_j}\}_{j=0,1,\ldots, p}$ in $O(n \log n)$ time, where the $O(n)$ factor pays for checking in $group_{\lambda^{(f)}_j}$ whether $i_\ell$ to $i_r$ are fused together with $i_k$ for $\lambda^{(f)}_j$. Then the total time to check $L_{j,k}$ for all $\lambda^{(f)}_j \geq \lambda^{(f)}_{j_0}$ is $O(n)$. Therefore the total time complexity to solve the inverse optimization problem is $O(n\log n + n) = O(n\log n)$.

The analysis and results in this section all apply to FL (\ref{prob: fused-lasso-lambda}) with $q = O(\frac{nU}{\epsilon})$. In particular, the path-of-solution algorithm, when applied to FL (\ref{prob: fused-lasso-lambda}), has time complexity $\tilde{O}(\frac{n^2U}{\epsilon})$, and the space complexity to store the path of solutions is $O(\frac{n^3U}{\epsilon})$.

\section{Conclusions}\label{sect: conclusion}
In this paper, we characterize the solution structure of the fused lasso problem FL (\ref{prob: fused-lasso-lambda}) of arbitrary convex loss functions as $\lambda$ varies and provide an algorithm to compute the path of solutions to FL for all $\lambda \geq 0$. The $\lambda$ parameter determines the relative importance between the loss terms and the regularization terms. Our method is to create an equivalent fused lasso problem PL-FL (\ref{prob: PL-FL-lambda}), to the solution accuracy $\epsilon$, with convex piecewise linear loss functions. The characterization and algorithm for the path of solutions to PL-FL are investigated, the results of which apply to FL of $\epsilon$ solution accuracy.

Besides being a bridge for FL of arbitrary convex loss functions, our results for PL-FL can also be applied to many problems in statistics, bioinformatics and signal processing where the loss functions are defined as convex piecewise linear functions in the first place. In those applications, finding a good value of $\lambda$ is a lengthy trial-and-error process. Our work makes the parameter tuning process more effective. If a large set/interval of pre-specified $\lambda$ values are to be examined, our algorithm is more efficient than solving PL-FL from scratch for every $\lambda$ value in the set/interval. In addition, our algorithm can efficiently solve the inverse optimization problem of finding a $\lambda$ value for the desired optimal solution, which makes design of experiments more effective.

\appendix
\section{Pseudo-code of {\sf compute\_$\lambda$\_breakpoint}}\label{appx: compute}
The pseudo-code {\sf compute\_$\lambda$\_breakpoint} computes potentially (at most one) new $\lambda$-breakpoint and (at most two) new maximal $\lambda$-constant-intervals in $G_{j,k}$. The optimal value of $I_{j,k}$ for the new maximal $\lambda$-constant-intervals is $a_{i_k,j_k}$. It follows the analysis of Lemma \ref{lem: m_breakpoints}, with a succinct presentation to summarize all cases discussed in the lemma.

\vspace{1em}

\noindent$(T_{j,I_{j,k}}, sink\_intv_j(I_{j,k})) := {\sf compute\_\lambda\_breakpoint}(I_{j,k}, n_j, \lambda^{(f)}_j, \lambda^{(f)}_{j+1}, smt_j(\cdot), a_{i_k, j_k}, T_{j,I_{j,k}}, sink\_intv_j(\cdot))$\\
\noindent\textbf{begin}\\
\taba \textbf{if} $I_{j,k} = 1$ \textbf{then} \{edge case\}\\
    \tabb \textbf{if} $smt_j(I_{j,k}) \geq 0$ \textbf{then}\\
        \tabc \textbf{if} $sink\_intv_j(I_{j,k} + 1) \neq \emptyset$ \textbf{then}\\
	    \tabd $\lambda_{k,\ell} := \max\{\lceil smt_j(I_{j,k})\rceil, \lambda^{(f)}_j\}$;\\
            \tabd $\lambda_{k,r} := sink\_intv_j(I_{j,k}+1).second$; \\
            \tabd $(T_{j, I_{j,k}}, sink\_intv_j(I_{j,k}))$ := {\sf update\_$\lambda$\_breakpoint}$(I_{j,k}, \lambda_{k,\ell}, \lambda_{k,r}, a_{i_k, j_k}, T_{j,I_{j,k}}, sink\_intv_j(I_{j,k}))$;
        \tabc \textbf{end if}\\
    \tabb \textbf{else} \{$smt_j(I_{j,k}) < 0$\} \\
        \tabc $\lambda_{k,\ell} := \lambda^{(f)}_j$;\\
        \tabc $\lambda_{k,r} := \max\{sink\_intv_j(I_{j,k} + 1).second, \min\{\lfloor tms_j(I_{j,k})\rfloor, \lambda^{(f)}_{j+1}-1\}\}$;\\
        \tabc  $(T_{j, I_{j,k}}, sink\_intv_j(I_{j,k}))$ := {\sf update\_$\lambda$\_breakpoint}$(I_{j,k}, \lambda_{k,\ell}, \lambda_{k,r}, a_{i_k, j_k}, T_{j,I_{j,k}}, sink\_intv_j(I_{j,k}))$;\\
    \tabb \textbf{end if}\\
\taba \textbf{else if} $I_{j,k} = n_j$ \textbf{then} \{edge case\}\\
    \tabb \textbf{if} $smt_j(I_{j,k}) \geq 0$ \textbf{then}\\
        \tabc \textbf{if} $sink\_intv_j(I_{j,k} - 1) \neq \emptyset$ \textbf{then}\\
	    \tabd $\lambda_{k,\ell} := \max\{\lceil smt_j(I_{j,k})\rceil, \lambda^{(f)}_j\}$;\\
            \tabd $\lambda_{k,r} := sink\_intv_j(I_{j,k} -1).second$; \\
            \tabd $(T_{j, I_{j,k}}, sink\_intv_j(I_{j,k}))$ := {\sf update\_$\lambda$\_breakpoint}$(I_{j,k}, \lambda_{k,\ell}, \lambda_{k,r}, a_{i_k, j_k}, T_{j,I_{j,k}}, sink\_intv_j(I_{j,k}))$;
        \tabc \textbf{end if}\\
    \tabb \textbf{else} \{$smt_j(I_{j,k}) < 0$\} \\
        \tabc $\lambda_{k,\ell} := \lambda^{(f)}_j$;\\
        \tabc $\lambda_{k,r} := \max\{sink\_intv_j(I_{j,k} - 1).second, \min\{\lfloor tms_j(I_{j,k})\rfloor, \lambda^{(f)}_{j+1}-1\}\}$;\\
        \tabc  $(T_{j, I_{j,k}}, sink\_intv_j(I_{j,k}))$ := {\sf update\_$\lambda$\_breakpoint}$(I_{j,k}, \lambda_{k,\ell}, \lambda_{k,r}, a_{i_k, j_k}, T_{j,I_{j,k}}, sink\_intv_j(I_{j,k}))$;\\
    \tabb \textbf{end if}\\
\taba \textbf{else} \{$1 < I_{j,k} < n_j$\}\\
    \tabb \textbf{if} $smt_j(I_{j,k}) \geq 0$ \textbf{then}\\
        \tabc \textbf{if} $sink\_intv_j(I_{j,k} - 1) \neq \emptyset$ and $sink\_intv_j(I_{j,k} + 1) \neq \emptyset$ \textbf{then}\\
	    \tabd $\lambda_{k,\ell} := \max\{\lceil smt_j(I_{j,k}) / 2\rceil, \lambda^{(f)}_j\}$;\\
            \tabd $\lambda_{k,r} := \min\{sink\_intv_j(I_{j,k} -1).second, sink\_intv_j(I_{j,k}+1).second\}$; \\
            \tabd $(T_{j, I_{j,k}}, sink\_intv_j(I_{j,k}))$ := {\sf update\_$\lambda$\_breakpoint}$(I_{j,k}, \lambda_{k,\ell}, \lambda_{k,r}, a_{i_k, j_k}, T_{j,I_{j,k}}, sink\_intv_j(I_{j,k}))$;
        \tabc \textbf{end if}\\
    \tabb \textbf{else} \{$smt_j(I_{j,k}) < 0$\} \\
        \tabc $\lambda_{k,\ell} := \lambda^{(f)}_j$;\\
        \tabc $\lambda_{k,r} := \max\big\{\max\{sink\_intv_j(I_{j,k} - 1).second, sink\_intv_j(I_{j,k}+1).second\},\\
        \tabf \min\{\lfloor tms_j(I_{j,k}) / 2\rfloor, \lambda^{(f)}_{j+1}-1\}\big\}$;\\
        \tabc  $(T_{j, I_{j,k}}, sink\_intv_j(I_{j,k}))$ := {\sf update\_$\lambda$\_breakpoint}$(I_{j,k}, \lambda_{k,\ell}, \lambda_{k,r}, a_{i_k, j_k}, T_{j,I_{j,k}}, sink\_intv_j(I_{j,k}))$;\\
    \tabb \textbf{end if}\\
\taba \textbf{end if}\\
\noindent\textbf{end}

\vspace{1em}

In the above pseudo-code, the subroutine {\sf update\_$\lambda$\_breakpoint} updates $T_{j,I_{j,k}}$ and $sink\_intv_j(I_{j,k})$ for the newly computed maximal $I_{j,k}$-sink-$\lambda$-interval $[\lambda_{k,\ell}, \lambda_{k,r}]$ (could be empty), from which (at most one) new $\lambda$-breakpoint and (at most two) new maximal $\lambda$-constant-intervals with optimal value $a_{i_k, j_k}$ for $I_{j,k}$ could be introduced. The pseudo-code is as follows:

\vspace{1em}

\noindent$(T_{j,I_{j,k}}, sink\_intv_j(I_{j,k})) := {\sf update\_\lambda\_breakpoint}(I_{j,k}, \lambda_{k,\ell}, \lambda_{k,r}, a_{i_k, j_k}, T_{j,I_{j,k}}, sink\_intv_j(I_{j,k}))$ \\
\noindent\textbf{begin}\\
    \taba \textbf{if} $\lambda_{k,\ell} \leq \lambda_{k,r}$ \textbf{then}\\
    	\tabb \textbf{if} $sink\_intv_j(I_{j,k}) = \emptyset$ \textbf{then}\\
	    \tabc $z := {\sf new\_node}(\lambda_{k,\ell}, \lambda_{k,r}, a_{i_k,j_k})$;\\
	    \tabc {\sf RB-INSERT}$(T_{j,I_{j,k}}, z)$;\\
	    \tabc $sink\_intv_j(I_{j,k}).first = \lambda_{k,\ell}, sink\_intv_j(I_{j,k}).second = \lambda_{k,r}$;\\
	\tabb \textbf{else}\\
    	    \tabc \textbf{if} $\lambda_{k,\ell} < sink\_intv_j(I_{j,k}).first$ \textbf{then} \\
        		\tabd $z := {\sf new\_node}(\lambda_{k,\ell}, sink\_intv_j(I_{j,k}).first-1, a_{i_k,j_k})$;\\
        		\tabd {\sf RB-INSERT}$(T_{j,I_{j,k}}, z)$;\\
        		\tabd $sink\_intv_j(I_{j,k}).first := \lambda_{k,\ell}$;\\
    	    \tabc \textbf{end if}\\
    	    \tabc \textbf{if} $\lambda_{k,r} > sink\_intv_j(I_{j,k}).second$ \textbf{then} \\
        		\tabd $z := {\sf new\_node}(sink\_intv_j(I_{j,k}).second+1, \lambda_{k,r}, a_{i_k,j_k})$;\\
	        \tabd {\sf RB-INSERT}$(T_{j,I_{j,k}}, z)$;\\
	        \tabd $sink\_intv_j(I_{j,k}).second := \lambda_{k,r}$;\\
    	    \tabc \textbf{end if}\\
	\tabb \textbf{end if}\\
    \taba \textbf{end if}\\
    \taba \textbf{return} $(T_{j,I_{j,k}}, sink\_intv_j(I_{j,k}))$;\\
\noindent\textbf{end}

\vspace{1em}

Recall that the number of nodes in each $T_{j,I_{j,k}}$ is $O(q)$. As a result, each call to {\sf RB-INSERT}$(T_{j,I_{j,k}}, z)$ is $O(\log q)$. Hence the complexity of {\sf update\_$\lambda$\_breakpoint} is $O(\log q)$. As a result, the complexity of {\sf compute\_$\lambda$\_breakpoint} is $O(\log q)$.

\end{document}